\documentclass[preprint,12pt]{elsarticle}

\usepackage{epsfig,epic,eepic,units}
\usepackage{hyperref}
\usepackage{url}
\usepackage{longtable}
\usepackage{mathrsfs}
\usepackage{multirow}
\usepackage{arydshln,leftidx,mathtools}
\usepackage{blkarray}
\usepackage{bigstrut}
\usepackage{amssymb}
\usepackage{graphicx}
\usepackage{subfig}
\usepackage{enumerate}	

\usepackage{amsmath}
\usepackage{amsfonts}
\usepackage[T1]{fontenc}

\usepackage{amsthm}
\DeclareMathOperator*{\argmax}{arg\,max}
\newtheorem{theorem}{Theorem}
\newtheorem{corollary}{Corollary}
\newtheorem{statement}{Statement}
\newtheorem{lemma}{Lemma}

\newtheorem{relation}{Relation}

\usepackage{algorithm}
\usepackage{algorithmic}

\usepackage{enumitem}
\usepackage{booktabs}
\usepackage{soul}

\newcommand{\U}{I\!\!H}  
\newcommand{\bu}{l\!\boldsymbol{h}}
\newcommand{\D}{\mathbb{C}}

\usepackage{multirow}
\usepackage{float}

\journal{Linear Algebra and its Applications}

\begin{document}

\begin{frontmatter}



\title{Markov Fundamental Tensor and Its Applications to Network Analysis}


\author{Golshan Golnari, Zhi-Li Zhang, and Daniel Boley}

\address{Department of Computer Science and Engineering, University of Minnesota}

\begin{abstract}
We first present a comprehensive review of various Markov
metrics used in the literature and express them in a consistent framework. We then introduce {\em fundamental
  tensor} -- a generalization of the well-known fundamental matrix --
and show that classical Markov metrics can be derived from it in
a unified manner.  We provide a collection of useful relations for Markov
metrics that are useful and insightful for network studies.  To
demonstrate the usefulness and efficacy of the proposed fundamental tensor in
network analysis,  we present four important applications: 1)
unification of network centrality measures, 2) characterization of
(generalized) network articulation points, 3) identification of
network most influential nodes, and 4) fast computation of network reachability after failures. 
\end{abstract}


\end{frontmatter}



\section{Introduction}
Random walk and Markov chain theory, which are in close relationship, shown to
be powerful tools in many fields from physics and chemistry to social sciences, economics,
and computer science \cite{kopp2012x,kutchukian2009fog,wilson1977quantum,acemoglu2011political,calvet2001forecasting}. For network analysis, too, they have shown
promises as effective tools \cite{he2002individual,ranjan2013geometry,golnari2015pivotality,chen2008clustering}, where the hitting time, a well-known Markov metric, is used to measure the distance (or similarity) between different parts of a network and provide more insight to structural properties of the network. We believe though that the applicability of Markov chain theory to network analysis is more extensive and is not restricted to using the hitting time. Markov chain theory enables us to provide more general solutions which cover the directed networks (digraphs) and is not tailored only to special case of undirected networks.   
 
In this paper, we revisit the fundamental matrix in Markov chain theory \cite{snell}, extend it to a more general form of tensor representation, which we call {\em fundamental tensor}, and use that to tackle four interesting network analysis applications. Fundamental tensor $\boldsymbol{F}_{smt}$ is defined \footnote{Note that the fundamental matrix is mostly denoted by $N$ in Markov chain theory literature, but since $N$ might reflect other meanings in computer science venues, we usually use $F$ (or $\boldsymbol{F}$ to denote the tensor) in our papers. } over three dimensions of source node $s$, middle (medial) node $m$, and target node $t$, which represents the expected number of times that the Markov chain visits node $m$ when started from $s$ and before hitting $t$ for the first time.
We remark that the (absorbing) fundamental matrix which is frequently referred to in this paper is different from (ergodic) fundamental matrix \cite{grinstead2012introduction} and is of special interest of the authors of this paper as: 1- it is nicely interpretable in terms of random walk, 2- it is conceptually interesting as aggregation over the (absorbing) fundamental tensor dimensions would result to different Markov metrics (Section \ref{sec:measureUnification}) and the articulation points of a network can be directly found from it (Section \ref{sec:articulpoint}), 2- it can be used for both applications that are modeled by an ergodic chain (Sections \ref{sec:measureUnification} and \ref{sec:articulpoint}) and applications modeled by an absorbing chain (Sections \ref{sec:most-influential} and \ref{sec:reachability}), 3- it is easily generalizable to absorbing Markov chain with multiple absorbing states (Section \ref{sec:target_set}) which we use to model network applications with multiple target nodes.
We show that computing the (absorbing) fundamental tensor is no harder than computing the (ergodic) fundamental matrix ($O(n^3)$) as the entire tensor can be computed very efficiently by a single matrix inversion and there is no need to compute the (absorbing) fundamental matrices for each absorbing state to form the entire tensor (which could take $O(n^4)$ of computations).

As the first application, we show that the fundamental tensor provides a unified way to compute the random walk distance (hitting time), random walk betweenness measure \cite{newman2005measure}, random walk closeness measure \cite{noh2004random}, and random walk topological index (Kirchhoff index)\cite{klein1993resistance} in a conceptual and insightful framework: hitting time distance as the aggregation of the fundamental tensor over the middle node dimension, betweenness as the aggregation over the source and target nodes, closeness as the aggregation over the source and middle node dimensions, and Kirchhoff index resulted as the aggregation over all the three dimensions.  
These four random walk measures are of well-known network analysis tools which have been vastly used in the literature \cite{boccaletti2006complex,blochl2011vertex,borgatti2005centrality,fortunato2010community,grady2006random,fouss2007random}.

In the second application, we extend
the definition of articulation points to the directed networks which has been originally defined
for undirected networks, known as cut vertices as well. We show that the (normalized) fundamental tensor nicely functions as a look up table to find all the articulation points of a directed network. Founded on the notion of articulation points, we also propose a load balancing measure for the networks. Load balancing is important
for network robustness against targeted attacks, where the balance in the loads help the network to show more resilience toward the failures. Through extensive experiments, we
evaluate the load balancing in several specific-shaped networks and real-world networks.

The applicability and efficiency of the fundamental tensor in social networks is the subject of the third application in this paper. We show that the (normalized) fundamental tensor can be used in the field of
social networks to infer the cascade and spread of a phenomena or an
influence in a network and derive a formulation to find the most
influential nodes for maximizing the influence spread over the
network. While the original problem is NP-hard, we propose a greedy algorithm which yields a provably near-optimal 
solution. We show that this algorithm outperforms the state-of-the-art as well as the centrality/importance measure baselines in maximizing the influence spread in the network. 

Since it is inefficient to use the regular reachability methods in large and dense networks with high
volume of reachability queries whenever a failure occurs in the network, devising an efficient dynamic reachability method is necessary in such cases. As the fourth application, we present a
dynamic reachability method in the form of a pre-computed oracle which is cable of answering
to reachability queries efficiently ($O(1)$) both in the case of having failures or no failure
in a general directed network. This pre-computed oracle is in fact the fundamental matrix computed for the extended network $G^o$ and target $o$. The efficiency of the algorithm is resulted from the theorem that we prove on incremental computation of the fundamental tensor when a failure happens in the network. The storage requirement of this oracle is only $O(n^2)$. 
Note that in the last two applications, the directed network $G$ does not need to be strongly connected, and our algorithms can be applied to any general network. 

For the sake of completeness, we also provide a comprehensive review of the other Markov metrics, such as hitting time, absorption probability, and hitting cost, which is a very useful metric for weighted networks and was introduced in a more recent literature \cite{fouss2007random}, but can be rarely found in Markov chain literature. In the review, we include Markov metrics' various definitions and formulations, and express them in a consistent form (matrix form, recursive form, and stochastic form). We also show that the fundamental tensor provides a basis for computing these Markov metrics in a unified manner. In addition, we review, gather, and derive many insightful relations for the Markov metrics.

The remainder of this paper is organized as follows. A preliminary on
network terminology is presented in Section~\ref{sec:prelim}.   In
Section~\ref{sec:definition}, we review and present various Markov metrics
in  a unified format. In Section~\ref{sec:relation}, we gather and
derive useful relations among the reviewed Markov metrics. Finally, four applications are presented in
Sections~\ref{sec:measureUnification},~\ref{sec:articulpoint},~\ref{sec:most-influential}, and~\ref{sec:reachability}  to demonstrate the usefulness and
efficacy of the fundamental tensor in network analysis.

\section{Preliminaries} \label{sec:prelim}
In general, a network can be abstractly modeled as a {\em weighted} and {\em
  directed} graph, denoted by
$G=(\mathcal{V},\mathcal{E},W)$. Here $\mathcal{V}$ is the set of
nodes in the network such as routers or switches in a communication
network or users in a social network, and its size is assumed to be
$n$ throughout the paper $|\mathcal{V}|=n$;  $\mathcal{E}$ is the set
of ({\em directed}) edges representing the (physical or logical)
connections between nodes ({\em e.g.}, a communication link from a node $i$
to a node $j$) or entity relations ({\em e.g.}, follower-followee relation
between two users). The  {\em affinity} (or adjacency) matrix
$A=[a_{ij}]$ is assumed to be nonnegative, {\em i.e.}, $a_{ij} \geq 0$,
where $a_{ij}>0$ if and only if edge $e_{ij}$ exists,
$e_{ij}\in\mathcal{E}$. The {\em weight} (or cost) matrix $W=[w_{ij}]$ represents the costs assigned to edges in a weighted network. Network $G$ is called strongly connected if
all nodes can be reachable from each other via at least one path. In this
paper, we focus on strongly connected networks, unless stated
otherwise.  

A random walk in $G$ is modeled by a discrete time Markov chain,
where the nodes of $G$ represent the states of the Markov chain. The target node in the network is modeled by an absorbing state at which the random walk arrives it cannot leave anymore.
The Markov chain is fully described by its transition probability
matrix: $P = D^{-1}A$, where $D$ is the diagonal matrix of
(out-)degrees, {\em i.e.}, $D=diag[d_i]$ and $d_i=\sum_{j}a_{ij}$. The $d_i$ is
often referred to as the (out-)degree of node $i$. Throughout the
paper, the words ``node" and ``state",  ``network" and ``Markov
chain" are often used interchangeably depending on the context. If the network $G$ is strongly
connected, the associated Markov chain is irreducible and the
stationary probabilities $\pi$ are strictly positive according to
Perron-Frobenius theorem \cite{gantmacher1960theory}. For an undirected
and connected $G$, the associated  Markov chain is reversible and the
stationary probabilities are a scalar multiple of node degrees:
$\pi_i=\frac{d_i}{\sum_i d_i}$.

\section{Definitions of Markov Metrics} \label{sec:definition}
 We review various Markov metrics and present them using three unified forms: 1) matrix form (and in terms of
 the fundamental matrix), 2) recursive form, and 3) stochastic form.
The matrix form is often the preferred form in this paper and we
 show how two other forms can be obtained from the matrix form. The stochastic form, however, provides a more intuitive definition
 of  random walk metrics. We also introduce {\em fundamental tensor}
 as a  generalization of the fundamental matrix and show how it can be computed efficiently.  
  
\subsection{Fundamental Matrix} 

The {\em expected number of visits}
counts the expected number of visits at a node, when a random walk
starts from a source node and before a stopping criterion. The
stopping criterion in random walk (or Markov) metrics is often ``visiting a target
node for the first time'' which is referred to as hitting
the target node. {\em Fundamental matrix} $F$ is formed for a specific target node, where the entries are the
expected number of visits at a medial node starting from a source
node, for all such pairs. In the following, the fundamental matrix 
is defined formally using three different forms. \footnote{The fundamental matrix that is referred to in this paper is defined for absorbing chain and is obtained from $F=(I-P_{11})^{-1}$. It is different from the fundamental matrix $Z=(I-P+\boldsymbol{1}\boldsymbol{\pi'})^{-1}$ which is defined for ergodic chain \cite{grinstead2012introduction}.}

\begin{itemize}

\item \textbf{Matrix form \cite{kemeny1960finite, snell}:}
Let $P$ be an $n \times n$ transition probability matrix for a strongly connected network $G$ and node $n$ be the target node. If the nodes are arranged in a way to assign the last index to the target node, transition probability matrix can be written in the form of $ P=\left[
\begin{array}{ c c }
P_{11} & \boldsymbol{p}_{12} \\
\boldsymbol{p}'_{21} & p_{nn}
\end{array} \right]$ and the fundamental matrix is defined as follows:
\begin{equation} \label{eq:N}
F=(I-P_{11})^{-1}, 
\end{equation}
where entry $F_{sm}$ represents the expected number of visits of medial node $m$, starting from source node $s$, and before hitting (or absorption by) target node $n$ \cite{snell}. Note that the target node can be any node $t$ which would be specified in the notation by $F^{\{t\}}$ to clarify that it is computed for target node $t$. 
This is discussed more in Markov metrics generalization to a set of targets (\ref{sec:target_set}).

Expanding $F_{sm}$ as a geometric series, namely,  $F_{sm}=[(I-P_{11})^{-1}]_{sm}=[I]_{sm}+[P_{11}]_{sm}+[P_{11}^2]_{sm}+...$, it is easy to see the probabilistic interpretation of the {\em expected number of visits} as a summation over the number of steps required to visit node $m$.

\item \textbf{Recursive form:} Each entry of the  fundamental matrix,
  $F_{sm}$,  can be recursively computed in terms of the entries of
  $s$'s outgoing neighbors. Note that if  $s=m$, $F_{sm}$ is
  increased by 1 to account for $X_0=m$ (the random walk starts at
  $s=m$, thus counting as the first visit at $m$).  

\begin{equation} \label{eq:F_recursive}
F_{sm}={1}_{\{s=m\}}+\sum_{j\in\mathcal{N}_{out}(s)} p_{sj}F_{jm} 
\end{equation}
It is easy to see the direct connection between the recursive form and
the matrix form: from $F=I+P_{11}F$, we have $F=(I-P_{11})^{-1}$.  

\item \textbf{Stochastic form \cite{norris1998markov}:}
Let $G={(X_k)}_{k>0}$ be a discrete-time Markov chain with the
transition probability matrix $P$, where $X_k$ is the state of Markov
chain in time step $k$. The indicator function
${1}_{\{X_k=m\}}$ is a Bernoulli random variable, equal to 1 if
the state of Markov chain is $m$ at time $k$, {\em i.e.} $X_k=m$, and 0
otherwise. The number of visits of node $m$, denoted by $\nu_m$, can
be written in terms of the indicator function: $\nu_m=\sum_{k=0}^{\infty}
{1}_{\{X_k=m\}}$. The stopping criteria is hitting target node $t$ for
the first time. In an irreducible chain, this event is guaranteed to
occur in a finite time. Hence $k<\infty$. $F_{sm}$ is defined as the
expected value of $\nu_m$ starting from $s$.

\begin{eqnarray}
F_{sm}&=&\mathbb{E}_s(\nu_m)=\mathbb{E}_s \sum_{k=0}^{<\infty} {1}_{\{X_k=m\}} = \sum_{k=0}^{<\infty} \mathbb{E}_s({1}_{\{X_k=m\}}) \nonumber
\\&=& \sum_{k=0}^{<\infty} \mathbb{P}(X_k=m|X_0=s,X_{<k}\neq t) = \sum_{k=0}^{<\infty} {[P^k_{11}]}_{sm}, 
\end{eqnarray}
where the expression  is simply the expanded version of the matrix form. 
Note that in order for  $F_{sm}$ to be finite (namely, the infinite
summation converges), it is sufficient that node $t$ be reachable from
all other nodes in network. In other words, the irreducibility of the entire network is not necessary.
\end{itemize}

\subsection{Fundamental Tensor}
We define the {\em fundamental tensor},
$\boldsymbol{F}$,  as a generalization of the fundamental matrix
$F^{\{t\}}$, which looks to be formed by stacking up the fundamental matrices constructed for
each node $t$ as the target node in a strongly connected network (Eq.(\ref{eq:fundamentalTensor})), but is in fact computed much more efficiently. In Theorem (\ref{thm:F_O3}), we show that the whole fundamental tensor can be computed from Moore-Penrose pseudo-inverse of Laplacian matrix with only $O(n^3)$ of complexity and there is no need to compute the fundamental matrices for every target node which require $O(n^4)$ of computation in total.
\begin{equation} \label{eq:fundamentalTensor}
\boldsymbol{F}_{smt}=
\begin{cases}
F^{\{t\}}_{sm}  & \text{if } s,m\neq t \\
0  & \text{if } s=t \text{ or } m=t
\end{cases}
\end{equation}
Fundamental tensor is presented in three dimensions of source node, medial (middle) node, and target node (Fig. (\ref{fig:tensor})).

\subsection{Hitting Time}

The (expected) {\em hitting time} metric, also known as the first
transit time, first passage time, and expected absorption time in the
literature, counts the expected number of steps (or time) required to
hit a target node for the first time when the random walk starts from
a source node. Hitting time is frequently used in the literature 
as a form of (random walk) distance metric for network analysis.  We
formally present it in three different forms below.

\begin{itemize}
\item \textbf{Matrix form} \cite{snell}:
Hitting time can be computed from the fundamental matrix (\ref{eq:N}) as follows:  
\begin{equation} \label{eq:H}
\boldsymbol{h}^{\{t\}}=F^{\{t\}}\boldsymbol{1},
\end{equation}
where $\boldsymbol{1}$ is a vector of all ones and $\boldsymbol{h}^{\{t\}}$ is a vector of $H_s^{\{t\}}$ computed for all $s \in \mathcal{V}\setminus \{t\}$. $H_s^{\{t\}}$ represents the expected number of steps required to hit node $t$ starting from $s$ and is obtained from: $H_s^{\{t\}}=\sum_{m} F^{\{t\}}_{sm}$. The intuition behind this formulation is that enumerating the average number of nodes visited on the way from the source node to the target node yields the number of steps (distance) required to reach to the target node. 


\item \textbf{Recursive form} \cite{grinstead2012introduction,norris1998markov,fouss2007random}:
The \textit{recursive} form of $H_s^{\{t\}}$ is the most well-known form presented in the literature for deriving the hitting time:
\begin{equation}
H_s^{\{t\}}=1+\sum_{m\in \mathcal{N}_{out}(s)}p_{sm}H_m^{\{t\}}
\end{equation}
It is easy to see the direct connection between the recursive form and
the matrix form: from
$\boldsymbol{h}=\boldsymbol{1}+P_{11}\boldsymbol{h}$, we have 
$\boldsymbol{h}=(I-P_{11})^{-1}\boldsymbol{1}$.

\item \textbf{Stochastic form} \cite{norris1998markov}:
Let $G={(X_k)}_{k>0}$ be a discrete-time Markov chain with the transition probability matrix $P$. The hitting time of the target node $t$ is denoted by a random variable $\kappa_{t}:\Omega\rightarrow\{0,1,2,...\}\cup\{\infty\}$
given by
$\kappa_{t}=\inf{\{\kappa\geq 0: X_{\kappa}=t\}}$,
where by convention the infimum of the empty set $\emptyset$ is
$\infty$. Assuming that the target node $t$ is reachable from all the
other nodes in the network,  we have $\kappa_t <\infty$. The
(expected) hitting time from $s$ to $t$ is then given by

\begin{eqnarray}
H_s^{\{t\}}&=&\mathbb{E}_s[\kappa_t]=\sum_{k=1}^{<\infty} k\mathbb{P}(\kappa_t=k|X_0=s)+\infty \mathbb{P}(\kappa_t=\infty|X_0=s) \nonumber
\\&=&\sum_{k=1}^{<\infty} k\mathbb{P}(X_k=t|X_0=s,X_{<k}\neq t) \nonumber
\\&=&\sum_{k=1}^{<\infty} k\sum_{m\neq t}\mathbb{P}(X_{k-1}=m|X_0=s,X_{<k-1}\neq t)\cdot\mathbb{P}(X_{k}=t|X_{k-1}=m) \nonumber
\\&=&\sum_{k=1}^{<\infty} k\sum_{m\neq t} {[P^{k-1}_{11}]}_{sm}{[\boldsymbol{p}_{12}]}_{m},
\end{eqnarray}
where $[P_{11}^0]_{sm}=1$ for $m=s$ and it is 0 otherwise. The
connection between the stochastic form and the matrix form can be
found in the appendix.

\end{itemize}

\subsubsection{Commute Time}

The commute time between node $i$ and node $j$ is defined as the sum
of the  hitting time from $i$ to $j$ and the hitting time from $j$ to $i$:
\begin{equation}
C_{ij}=H_i^{\{j\}}+H_j^{\{i\}}
\end{equation}
Clearly, commute time is a symmetric quantity,
{\em i.e.,} $C_{ij}=C_{ji}$. In contrast,  hitting time is in general not
symmetric, even when   the network is undirected. 

\subsection{Hitting Cost}
The (expected) {\em hitting cost}, also known as average first-passage
cost in the literature, generalizes the (expected) hitting time by
assigning a cost to each transition. Hitting cost from $s$ to $t$, denoted by $\U_s^{\{t\}}$, is
the average {\em cost} incurred by the random walk starting
from node $s$ to hit node $t$ for the first time. The cost of
transiting edge $e_{ij}$ is given by $w_{ij}$. The hitting cost was
first introduced by Fouss {\em et al.} \cite{fouss2007random} and given in
a recursive form. In the following, we first provide a rigorous definition for hitting cost
in the stochastic form, and then show how the matrix form and recursive form can be driven from this definition. 

\begin{itemize}
	
	\item \textbf{Stochastic form:}
	Let $G={(X_k)}_{k>0}$ be a discrete-time Markov chain with the
        transition probability matrix $P$ and cost matrix $W$.	The
        hitting cost of the target node $t$ is a random variable $\eta_{t}:\Omega\rightarrow \mathcal{C}$ which is defined by
	$\eta_{t}=\inf{\{\eta\geq 0: \exists k, X_k=t, \sum_{i=1}^k
          w_{X_{i-1}X_i}=\eta\}}$. $\mathcal{C}$ is a countable set. 
  If we view $w_{ij}$ as the length of edge (link) $e_{ij}$, then the
  hitting cost $\eta_{t}$ is the total length of steps that the random
  walk takes until it hits $t$ for the first time.  The expected value
  of $\eta_{t}$ when the random walk starts at node $s$ is given by
	\begin{equation} \label{eq:stoch_B}
	\U_{s}^{\{t\}}=\mathbb{E}_s[\eta_t]=\sum_{l\in\mathcal{C}} l\mathbb{P}(\eta_t=l|X_0=s)
	\end{equation}
	For compactness, we delegate the more detailed derivation of the stochastic form and its connection
        with the matrix form to the appendix.

\item \textbf{Matrix form:}
	Hitting cost can be computed from the following \textit{closed form} formulation:
	\begin{equation} \label{eq:B=Fr}
	\bu^{\{t\}}=F \boldsymbol{r},
	\end{equation}
	where $\boldsymbol{r}$ is the vector of expected outgoing costs and $\bu^{\{t\}}$ is a vector of $\U_s^{\{t\}}$ computed for all $s \in \mathcal{V}\setminus \{t\}$. The expected outgoing cost of node $s$ is obtained from: $r_s=\sum_{m\in \mathcal{N}_{out}(s)} p_{sm}w_{sm}$. Note
	that the hitting time matrix $H$ in Eq.(\ref{eq:H}) is a
        special case of the hitting cost matrix $\U$, obtained when $w_{ij}=1$ for all $e_{ij}$.

	\item \textbf{Recursive form \cite{fouss2007random}:}
	The recursive computation of $\U_{s}^{\{t\}}$ is given as follows:
	\begin{equation} \label{eq:recursive_B}
	\U_{s}^{\{t\}}=r_s + \sum_{m\in \mathcal{N}_{out}(s)}p_{sm}\U_{m}^{\{t\}}. 
	\end{equation}	
	It is easy to see the direct connection between the recursive
        form and the matrix form: from $\bu=\boldsymbol{r}+P_{11}\bu$,
        we have $\bu=(I-P_{11})^{-1}\boldsymbol{r}$.

\end{itemize}

\subsubsection{Commute Cost}
Commute cost $\D_{ij}$ is defined as the expected cost required to hit
$j$ for the first time and get back to $i$. As in the case of commute
time, commute cost is a symmetric metric and is given by
\begin{equation}
\D_{ij}=\U_{i}^{\{j\}}+\U_{j}^{\{i\}}
\end{equation} 

\subsection{Absorption Probability}

The absorption probability, also known as hitting probability in the
literature, is the probability of hitting or getting absorbed by a
target node (or  any node in a set of  target nodes) in a finite time~\cite{norris1998markov}. For a
single target node, this probability is trivially equal to 1 for all
nodes in a strongly connected network. We therefore consider  more
than one target nodes in this paper.   

Let indexes $n-1$ and $n$ be assigned to two target nodes in a
strongly connected network. We partition the transition probability
matrix $P$ as follows:
%
%
%
%

\begin{equation}
P=\begin{blockarray}{cccc}
& n-1 & n & \\
\begin{block}{[ccc]c}
P_{11} & \boldsymbol{p}_{12} &  \boldsymbol{p}_{13} &  \\
\boldsymbol{p}'_{21} & p_{n-1,n-1} & p_{n-1,n} & n-1\\
 \boldsymbol{p}'_{31} & p_{n,n-1} & p_{n,n} & n  \\
\end{block}
\end{blockarray}
\end{equation} 
where $P_{11}$ is an $(n-2)\times(n-2)$ matrix, $\boldsymbol{p}_{12}$, $\boldsymbol{p}_{13}$, $\boldsymbol{p}_{21}$, and $\boldsymbol{p}_{31}$ are $(n-2)\times1$ vectors, and the rest are scalars.
The corresponding absorption probability can be expressed in three
forms as follows:
\begin{itemize}
	\item \textbf{Matrix form \cite{snell}:} The absorption probability matrix denoted by $Q$ is a $(n-2)\times 2$ matrix whose columns represent the absorption probabilities to target $n-1$ and $n$ respectively:
	
	\begin{eqnarray} \label{eq:Q}
	Q^{\{n-1,\overline{n}\}}&=&F\boldsymbol{p}_{12},
	\\Q^{\{\overline{n-1},n\}}&=&F\boldsymbol{p}_{13},
	\end{eqnarray}
	where $F=(I-P_{11})^{-1}$. The notation
        $Q^{\{n-1,\overline{n}\}}$ emphasizes that target $n-1$ is hit
        sooner than target $n$, and $Q^{\{\overline{n-1},n\}}$
        indicates hitting target $n$ occurs sooner than target
        $n-1$. The formulation above states that to obtain the
        probability of getting absorbed (hit) by a given target when
        starting a random walk from a source node, we
        add up the absorption probabilities of starting from the
        neighbors of the source node, weighted by the number of times
        we expect to be in those neighboring nodes \cite{snell}. For a
        strongly connected network, these two probabilities are sum up
        to 1 for each starting node $s$, {\em i.e.,} $Q^{\{n-1,\overline{n}\}}_s+Q^{\{\overline{n-1},n\}}_s=1$. 
	
	\item \textbf{Recursive form \cite{norris1998markov}:} For each of the target nodes, the absorption probability starting from any source node can be found from the absorption probabilities starting from its neighbors:
	\begin{eqnarray} 
	Q^{\{\overline{n-1},n\}}_s&=&\sum_{m\in \mathcal{N}_{out}(s)}p_{sm}Q^{\{\overline{n-1},n\}}_m,
	\\Q^{\{n-1,\overline{n}\}}_s&=&\sum_{m\in \mathcal{N}_{out}(s)}p_{sm}Q^{\{n-1,\overline{n}\}}_m,
	\end{eqnarray}
	where $s,m \in \mathcal{V}\setminus\{n-1,n\}$.
	Note that the neighbors of a node can also be the target
        nodes. Thus, the right-hand side of the above equations is
        decomposed into two parts: $Q^{\{\overline{n-1},n\}}_s= p_{sn}
        + \sum_{m\neq n,n-1}p_{sm}Q^{\{\overline{n-1},n\}}_m$, and the
        same way for $Q^{\{n-1,\overline{n}\}}_s$. Now, it is easy to
        see how the recursive form is connected to the matrix form:
        from
        $Q^{\{\overline{n-1},n\}}=\boldsymbol{p}_{13}+P_{11}Q^{\{\overline{n-1},n\}}$,
        we have 
$Q^{\{\overline{n-1},n\}}=(I-P_{11})^{-1}\boldsymbol{p}_{13}$.
	 
	\item \textbf{Stochastic form \cite{norris1998markov}:}
	Let $G={(X_k)}_{k>0}$ be a discrete-time Markov chain with the
        transition matrix $P$. The hitting time of the target state $n$ before $n-1$ is denoted by a random variable $\kappa_{n}:\Omega\rightarrow\{0,1,2,...\}\cup\{\infty\}$
	given by
	$\kappa_{n}=\inf{\{\kappa\geq 0: X_{\kappa}=n,
          X_{k<\kappa}\neq n,n-1\}}$. Then the probability of ever
        hitting $n$ is $\mathbb{P}(\kappa_{n}<\infty)$
        \cite{norris1998markov}. This can be derived as follows:
	\begin{eqnarray}
	Q^{\{\overline{n-1},n\}}_s &=& \sum_{k=1}^{<\infty} \mathbb{P}(\kappa_n=k|X_0=s) \nonumber
	\\&=&\sum_{k=1}^{<\infty} \mathbb{P}(X_k=n|X_0=s,X_{<k}\neq n,n-1) \nonumber
	\\&=&\sum_{k=1}^{<\infty} \sum_{m\neq n, n-1}\mathbb{P}(X_{k-1}=m|X_0=s,X_{<k-1}\neq n,n-1)\cdot\nonumber\\&&\mathbb{P}(X_{k}=n|X_{k-1}=m) \nonumber
	\\&=&\sum_{k=1}^{<\infty} \sum_{m} {[P^{k-1}_{11}]}_{sm}{[\boldsymbol{p}_{13}]}_{m} \nonumber
	\\&=&\sum_{m}
              {[(I-P_{11})^{-1}]}_{sm}{[\boldsymbol{p}_{13}]}_{m}.
	\end{eqnarray}
The stochastic form for $Q^{\{n-1,\overline{n}\}}_s$ is derived in a
similar vein.

\end{itemize}

\subsection{Generalization: Markov Metrics for a Set of Targets} \label{sec:target_set}

Let $\mathcal{A}=\{t_1,...,t_{|\mathcal{A}|}\}$ be a set of target nodes. Then the transition probability matrix can be written in the following form:
\begin{equation}
P=\left[
\begin{array}{ c c }
P_{\mathcal{TT}} & P_{\mathcal{TA}} \\
P_{\mathcal{AT}} & P_{\mathcal{AA}}
\end{array} \right],
\end{equation}
where $\mathcal{T}=\mathcal{V}\setminus \mathcal{A}$ is the set of
non-target nodes. Note that set of target nodes
$\mathcal{A}$ can  be modeled as the set of absorbing states in a
Markov chain, and then $\mathcal{T}=\mathcal{V}\setminus \mathcal{A}$ is the set of
transient (non-absorbing) nodes. Since hitting the target nodes is the stopping criterion
for all the Markov metrics we have reviewed so far,  it does not
matter where the random walk can go afterwards and what the outgoing
edges of the target nodes are. Therefore, there is no difference
between $P=\left[ 
\begin{array}{ c c }
P_{\mathcal{TT}} & P_{\mathcal{TA}} \\
P_{\mathcal{AT}} & P_{\mathcal{AA}}
\end{array} \right]$ and $P=\left[
\begin{array}{ c c }
P_{\mathcal{TT}} & P_{\mathcal{TA}} \\
\textbf{0} & I
\end{array} \right]$ for computing the Markov metrics.

For a given set of target nodes $\mathcal{A}$, the fundamental matrix
$F^{\mathcal{A}}$ is obtained using the following relation: 
\begin{equation}
F^{\mathcal{A}}=I+\sum_{k=1}^{<\infty} P^k_{\mathcal{TT}}=(I-P_{\mathcal{TT}})^{-1},
\end{equation}
which is a  general form of the fundamental matrix defined for a
single target (Eq.(\ref{eq:N})). Entry $F^{\mathcal{A}}_{sm}$
represents the expected number of visits to $m$ before hitting
\textit{any} of the target nodes in $\mathcal{A}$ when starting a
random walk from $s$. 

A hitting time for $\mathcal{A}$ is defined as the expected number of
steps to hit the set for the first time which can occur by hitting
{\em any} of the target nodes in this set. The vector of hitting
times with respect to a target set $\mathcal{T}$ can be computed
using 
\begin{equation}
\boldsymbol{h}^{\mathcal{A}}=F^{\mathcal{A}}\boldsymbol{1}
\end{equation}

If there exists a matrix of costs $W$ defined for the network, the
hitting cost for target set $\mathcal{A}$ is given below
\begin{equation}
\bu^{\mathcal{A}}=F^{\mathcal{A}}\boldsymbol{r},
\end{equation}
where $\boldsymbol{r}$ is a vector of expected outgoing cost $r_s$'s: $r_s=\sum_{m\in \mathcal{N}_{out}(s)} p_{sm}w_{sm}$.

The absorption probability of target set $\mathcal{A}$ is a $|\mathcal{T}|\times|\mathcal{A}|$ matrix whose columns represents the absorption probability for each target node if it gets hit sooner than the other target nodes:
\begin{equation}
Q^{\mathcal{A}}=F^{\mathcal{A}}P_{\mathcal{TA}},
\end{equation}
where $Q^{\mathcal{A}}$ is a row-stochastic matrix for a strongly connected network.

We remark that if the network is not strongly connected (thus the
corresponding Markov chain is not irreducible), $I-P_{\mathcal{TT}}$
may not be  non-singular for every set of $\mathcal{A}$. Hence
$F^{\mathcal{A}}$ may not exist. The necessary and sufficient
condition for the existence of $F^{\mathcal{A}}$ is that target set $\mathcal{A}$ includes {\em at
least one node from each recurrent equivalence class} in the
network. The recurrent equivalence
class is the minimal set of nodes that have no outgoing
edge to nodes outside the set. Once a random walk reaches a node in 
a recurrent equivalence class, it can no longer get out of that set. A recurrent equivalence class can be as small as
one single node, which is called an absorbing node.

\section{Useful Relations for Markov Metrics} \label{sec:relation}
In this section, we first establish several important theorems, and then
gather and derive a number of useful relations among the Markov
metrics. We start by relating the fundamental tensor with the Laplacian
matrices of a general network. For an undirected network or graph $G$, the
graph Laplacian $L^{u}=D-A$ (where $A$ is the adjacent matrix of $G$ and
$D=diag[d_i]$ is the diagonal matrix of node degrees) and its normalized
version  $\tilde{L}^{u}=D^{-\frac{1}{2}}LD^{-\frac{1}{2}}$ have been
widely studied and found many applications (see, {\em e.g.},~\cite{chung94spectral} and the references therein). In particular, it is
well-known that  commute times are closely related to the
Penrose-Moore pseudo-inverse of $L^{u}$ (and a variant of Eq.(\ref{eq:C-L})
also holds for $\tilde{L}^{u}$):
	\begin{equation} 
	C_{ij}=L^{u,+}_{ii}+L^{u,+}_{jj} -L^{u,+}_{ij}-L^{u,+}_{ji}. \label{eq:C-L}
	\end{equation}

Li and Zhang~\cite{Yanhua-Infocom10,WAW2010Yanhua,Li:Digraph} were first to
introduce the (correct) generalization of graph Laplacians for directed
networks/graphs ({\em digraphs}) using the stationary distribution 
$\{\pi_i\}$ of the transition matrix $P=D^{-1}A$ for the associated (random walk)
Markov chain defined on a directed network $G$. For a strongly
connected network $G$, its  {\em normalized} {\em digraph Laplacian} is
defined as $\tilde{L} =\Pi^{\frac{1}{2}} (I-P) \Pi^{-\frac{1}{2}}$, where $\Pi=diag[\pi_i]$ is the diagonal
matrix of stationary probabilities. Li and Zhang proved that
the hitting time and commute time can be computed from the Moore-Penrose pseudo-inverse $\tilde{L}^{+}$ of $\tilde{L}$ using the following relations:
	\begin{equation} 
	H_{i}^{\{j\}}=\frac{\tilde{L}^{+}_{jj}}{\pi_j} - \frac{\tilde{L}^{+}_{ij}}{\sqrt{\pi_i\pi_j}} \label{eq:H-diL}
	\end{equation}
and 
	\begin{equation} 
C_{ij}=H_{i}^{\{j\}}+ H_{j}^{\{i\}}=\frac{\tilde{L}^{+}_{ii}}{\pi_i}+\frac{\tilde{L}^{+}_{jj}}{\pi_j} - \frac{\tilde{L}^{+}_{ij}}{\sqrt{\pi_i\pi_j}} -\frac{\tilde{L}^{+}_{ji}}{\sqrt{\pi_i\pi_j}}.\label{eq:C-diL}
	\end{equation}

We define the ({\em unnormalized}) {\em digraph Laplacian} for a general
(directed or undirected) network $G$ as $L=\Pi(I-P)$ and the {\em
  random walk Laplacian} as $L^{p}=I-P$. Clearly, $\tilde{L}
=\Pi^{-\frac{1}{2}} L \Pi^{-\frac{1}{2}} =\Pi^{\frac{1}{2}}
L^{p}\Pi^{-\frac{1}{2}}$.  Note that for a (connected) undirected
graph, as $\pi_i=\frac{d_i} {vol(G)}$ where $vol(G)=\sum_j d_j$, we see that
the classical graph Laplacian $L^{u}=D-A =  vol(G)L$. Any results which hold
for $L$ also hold for $L^{u}=D-A$ with a scalar multiple. In the following we relate
the fundamental tensor to the digraph and random walk Laplacians $L$
and $L^{p}$, and use this relation to establish similar expressions for
computing hitting and commute times using $L$, analogous to
Eqs.(\ref{eq:H-diL}) and (\ref{eq:C-diL}).

\begin{lemma}[\cite{boley2011commute}] \label{lemma:boley}
	Let $ \left[
	\begin{array}{ c c }
	L_{11} & \boldsymbol{l}_{12} \\
	\boldsymbol{l}'_{21} & l_{nn}
	\end{array} \right]$ be an $n \times n$ irreducible matrix
      such that $nullity(L)=1$. Let $M=L^+$ be the Moore-Penrose pseudo-inverse of $L$
      partitioned similarly and $(\mathbf{u}',1)L=0$, $L(\mathbf{v};1)=0$, where
      $\mathbf{u}$ and $\mathbf{v}$ are $(n-1)$-dim column vectors, $\mathbf{u}'$ is
      the transpose of the column vector $\mathbf{u}$
      ($(\mathbf{u}',1)$ is a $n$-dim row vector and $(\mathbf{v};1)$
      is a $n$-dim column vector, {\em a la} MATLAB). Then the inverse of the $(n-1) \times (n-1)$ matrix $L_{11}$ exists and is given by:
	\begin{equation} 
	L_{11}^{-1}=(I+\mathbf{v}\mathbf{v}')M_{11}(I+\mathbf{u}\mathbf{u}'),
	\end{equation}
	where $I$ denotes the $(n-1)\times (n-1)$ identity matrix.
\end{lemma}
Note that node $n$ in the above lemma can be substituted by any other node (index).  

\begin{theorem} \label{thm:F_O3}
	The fundamental tensor can be computed from the Moore-Penrose
        pseudo-inverse of the {\em digraph Laplacian matrix}
        $L=\Pi(I-P)$ as well as the  {\em random
          walk} Laplacian matrix
        ${L^p}=I-P$ as follows, which results to $O(n^3)$ time complexity:
	\begin{equation} \label{N+}
	\boldsymbol{F}_{smt}=(L_{sm}^+ - L_{tm}^+ + L_{tt}^+ - L_{st}^+)\pi_m,
	\end{equation}
	\begin{equation} \label{Lp+}
	\boldsymbol{F}_{smt}={L^p}_{sm}^+ - {L^p}_{tm}^+ + \frac{\pi_m}{\pi_t}{L^p}_{tt}^+ - \frac{\pi_m}{\pi_t}{L^p}_{st}^+,
	\end{equation}
	where $\pi_i$ is the stationary probability of node $i$ and $\Pi$ is a diagonal matrix whose $i$-th diagonal entry is equal to $\pi_i$.
\end{theorem} 
\begin{proof}
 Note that $F=(I-P_{11})^{-1}=L_{11}^{-1}$ as in
 Lemma~\ref{lemma:boley}.  The above equations follow from
 Lemma~\ref{lemma:boley} with $\boldsymbol{v}=\boldsymbol{u}=\boldsymbol{1}$.
 The nullity of matrix $L^{p}=I-P$ for a strongly connected network
 is 1. Using Eq.(\ref{N+}) or (\ref{Lp+}), all $n^3$ entries of the fundamental tensor $\boldsymbol{F}$ can be computed from $L^+$ in constant time each.
\end{proof}

\begin{corollary}  \label{cor:F}
	\begin{equation}
	\sum_{s,t} \boldsymbol{F}_{smt}=c\pi_m,
	\end{equation}
	where $c$ is a constant independent of $m$.
\end{corollary}
\begin{proof}
	\begin{eqnarray}  \label{eq:sumF=pi}
	\sum_{s,t} \boldsymbol{F}_{smt} &=& \sum_{s,t} (L_{sm}^+ - L_{tm}^+ - L_{st}^+ + L_{tt}^+)\pi_m  
	\\&=& 0 - 0 - 0 + (n \sum_t L_{tt}^+)\pi_m
	\\&=& c\pi_m,
	\end{eqnarray}
	where the second equality follows from the fact that the
        column sum of $L^+=(\Pi(I-P))^+$ is zero. Later in
        Section~\ref{sec:measureUnification}, we will show that
        $c=|\mathcal{E}|K$, where $K$ is the Kirchhoff index of a
        network. 
\end{proof}

\begin{corollary}
	Hitting time and commute time can also be expressed in terms
        of entries in the digraph Laplacian matrix $L=\Pi(I-P)$~\cite{Li:Digraph}:
	\begin{equation} \label{H+}
	H_i^{\{j\}}=\sum_m (L_{im}^+ - L_{jm}^+)\pi_m + L_{jj}^+ - L_{ij}^+,
	\end{equation}
	\begin{equation} \label{C+}
	C_{ij}=L_{ii}^+ + L_{jj}^+ - L_{ij}^+ - L_{ji}^+,
	\end{equation}
\end{corollary}
\begin{proof}
	Use Eq.(\ref{eq:H}) and (\ref{N+}).
\end{proof}
Note that we can also write the metrics in terms of the random walk Laplacian matrix ${L^p}$ by a simple substitution: $L_{im}^+ - L_{ij}^+ = \frac{{L^p}_{im}^+}{\pi_m} - \frac{{L^p}_{ij}^+}{\pi_j}$.

\begin{corollary}
	Hitting cost $\U$ and commute cost $\D$ can be expressed  in terms of the digraph Laplacian matrix $L=\Pi(I-P)$:
	\begin{equation} \label{U+}
	\U_{ij}=\sum_m (L_{im}^+ - L_{jm}^+ + L_{jj}^+ - L_{ij}^+)g_m,
	\end{equation}
	\begin{equation} \label{Y+}
	\D_{ij}=(L_{im}^+ - L_{jm}^+ + L_{jj}^+ - L_{ij}^+)\sum_m g_m,
	\end{equation}
	where $g_m=r_m \pi_m$ and $r_m=\sum_{k\in \mathcal{N}_{out}(m)} p_{mk}w_{mk}$.
\end{corollary}
\begin{proof}
	Use Eq.(\ref{eq:B=Fr}) and (\ref{N+}). From Eq.(\ref{C+})
        and (\ref{Y+}), it is also interesting to note that commute cost is a multiple scalar of commute time.
\end{proof}

\begin{lemma}[\cite{boley2011commute}] \label{lemma:boley2}
	Let $C$ be an $n\times n$ non-singular matrix and suppose $A=C-\boldsymbol{uv}'$ is singular. Then the Moore-Penrose pseudo-inverse of $A$ is given as:
	\begin{equation}
	A^+=(I-\frac{\boldsymbol{xx'}}{\boldsymbol{x'x}})C^{-1}(I-\frac{\boldsymbol{yy'}}{\boldsymbol{y'y}}),
	\end{equation}
	where $\boldsymbol{x}=C^{-1}\boldsymbol{u}$, $\boldsymbol{y'}=\boldsymbol{v'}C^{-1}$.
\end{lemma}

\begin{theorem} \label{thm:Z}
	For an ergodic Markov chain, the Moore-Penrose pseudo-inverse of random-walk Laplacian ${{L^p}}^+$ can be computed from fundamental matrix $Z=(I-P+\boldmath{1}\boldsymbol{\pi'})^{-1}$ \cite{grinstead2012introduction} as follows:
	\begin{equation} \label{eq:LpZ}
	{L^p}^+=(I-\frac{Z\boldsymbol{1}\boldsymbol{1'}Z'}{\boldsymbol{1'}Z'Z\boldsymbol{1}})Z(I-\frac{Z'\boldsymbol{\pi}\boldsymbol{\pi'}Z}{\boldsymbol{\pi'}ZZ'\boldsymbol{\pi}}),
	\end{equation}
	where $\boldsymbol{1}$ is a vector of all 1's and $\boldsymbol{\pi}$ denotes the vector of stationary probabilities.
\end{theorem}
\begin{proof}
	The theorem is a direct result of applying Lemma \ref{lemma:boley2}.
\end{proof}
Theorem \ref{thm:Z} along with Theorem \ref{thm:F_O3} reveal the relation between the fundamental matrices $F$ and $Z$. They also show that the fundamental tensor $F$ can be computed by a single matrix inverse, can it be either a Moore-Penrose pseudo-inverse or a regular matrix inverse, as ${L^p}^+$ in Eq.~(\ref{Lp+}) can be computed by either operating the pseudo-inverse on $L^p$ or using Eq.~(\ref{eq:LpZ}). 
Discussion on computing Markov metrics via the group inverse can be found in \cite{meyer1975role,kirkland2012group}.
 
\begin{theorem}[Incremental Computation of the Fundamental Matrix] \label{Nlemma1}
	The fundamental matrix for  target set  $\mathcal{S}_1\cup
        \mathcal{S}_2$ can be computed from the fundamental matrix for
        target set $\mathcal{S}_1$ as follows,
	\begin{equation} 
	F^{\mathcal{S}_1\cup \mathcal{S}_2}_{im}=F^{\mathcal{S}_1}_{im}-F^{\mathcal{S}_1}_{i\mathcal{S}_2}[{F^{\mathcal{S}_1}_{\mathcal{S}_2\mathcal{S}_2}}]^{-1}F^{\mathcal{S}_1}_{\mathcal{S}_2m},
	\end{equation}
	where  $F^{\mathcal{S}_1}_{i\mathcal{S}_2}$ denotes the row
        corresponding to node $i$ and the columns corresponding to set
        $\mathcal{S}_2$ of the fundamental matrix $F^{\mathcal{S}_1}$,
        and the (sub-)matrices
        ${F^{\mathcal{S}_1}_{\mathcal{S}_2\mathcal{S}_2}}$ and
        $F^{\mathcal{S}_1}_{\mathcal{S}_2m}$ are similarly defined.
\end{theorem}


\begin{proof}
	Consider the  matrix $M=I-P_{\mathcal{TT}}$,  where the
        absorbing set is $\mathcal{A}=\mathcal{S}_1$ and the transient
        set $\mathcal{T}=V\setminus \mathcal{S}_1$. The inverse of $M$
        yields the fundamental matrix $F^{\mathcal{S}_1}$, and the
        inverse of its sub-matrix obtained from removing rows and
        columns corresponding to set $\mathcal{S}_2$ yields the
        fundamental matrix $F^{\mathcal{S}_1\cup
          \mathcal{S}_2}$. Using the
        following equations from the Schur complement,  we see that the inverse of  a
        sub-matrix can be derived 
        from that of the original matrix.
	
 If $A$ is invertible, we can factor the matrix $M=\left[
	\begin{array}{ c c }
	A & B \\
	C & D
	\end{array} \right]$ as follows
	\begin{equation}
	\left[
	\begin{array}{ c c }
	A & B \\
	C & D
	\end{array} \right] =
	\left[
	\begin{array}{ c c }
	I & 0 \\
	CA^{-1} & I
	\end{array} \right]
	\left[
	\begin{array}{ c c }
	A & B \\
	0 & D-CA^{-1}B
	\end{array} \right]
	\end{equation}
	Inverting both sides of the equation yields
	\begin{eqnarray}
	\left[
	\begin{array}{ c c }
	A & B \\
	C & D
	\end{array} \right]^{-1} &=&
	\left[
	\begin{array}{ c c }
	A^{-1} & -A^{-1}BS^{-1} \\
	0 & S^{-1}
	\end{array} \right]
	\left[
	\begin{array}{ c c }
	I & 0 \\
	-CA^{-1} & I
	\end{array} \right]
	\\&=& 
	\left[
	\begin{array}{ c c }
	A^{-1}+A^{-1}BS^{-1}CA^{-1} & -A^{-1}BS^{-1} \\
	-S^{-1}CA^{-1} & S^{-1}
	\end{array} \right]
	\\&=& 
	\left[
	\begin{array}{ c c }
	X & Y \\
	Z & W
	\end{array} \right],
	\end{eqnarray}
	where $S=D-CA^{-1}B$. Therefore, $A^{-1}$ can be computed from $A^{-1}=X-YW^{-1}Z$.
\end{proof}

\begin{corollary}
	The simplified form of Theorem~\ref{Nlemma1} for a single
        target is given by
	\begin{equation} \label{Nlemma1_simple}
	F^{\{j,k\}}_{im}=F^{\{j\}}_{im}-\frac{F^{\{j\}}_{ik}F^{\{j\}}_{k,m}}{F^{\{j\}}_{k,k}}
	\end{equation}
\end{corollary}

\begin{lemma} \label{Nlemma2general}
	\begin{equation} 
	P_{\mathcal{T}\mathcal{T}}F^{\mathcal{A}}=F^{\mathcal{A}}P_{\mathcal{T}\mathcal{T}}=F^{\mathcal{A}},  
	\end{equation}
	where $\mathcal{T}\cup \mathcal{A}=\mathcal{V}$
\end{lemma}
\begin{proof}
	It follows easily from Eq.(\ref{eq:N}).
\end{proof}

\begin{corollary}[Another Recursive Form for the Fundamental Matrix]
	\begin{equation} \label{Nlemma2}
	F^{\{j\}}_{im}=
	\begin{cases}
	\sum_{k\in\mathcal{N}_{in}(m)} F^{\{j\}}_{ik}p_{km}  & \text{if } i\neq m \\
	1+\sum_{k\in\mathcal{N}_{in}(m)} F^{\{j\}}_{ik}p_{km}  & \text{if } i= m 
	\end{cases}
	\end{equation}	  
\end{corollary}
\begin{proof}
	It is a special case of Lemma~\ref{Nlemma2general}. Note that
        the recursive relation in Eq.(\ref{eq:F_recursive}) is  in terms of $s$'s outgoing neighbors, while this one is in terms of incoming neighbors of $m$.
\end{proof} 

\begin{theorem}[Absorption Probability \& Normalized Fundamental Matrix] \label{Nlemma3}
	The absorption probability of a target node $j$ in an
        absorbing set $\mathcal{A}=\{j\}\cup \mathcal{S}$ can be
        written in terms of the normalized fundamental matrix
        $F^\mathcal{S}$,  where the columns are normalized by the diagonal entries:
	\begin{equation} 
	Q^{\mathcal{A}}_{ij}=\frac{F^{\mathcal{S}}_{ij}}{F^{\mathcal{S}}_{jj}}
	\end{equation}
\end{theorem}
\begin{proof}
	\begin{eqnarray}
	Q^{\mathcal{A}}_{ij} &=&[F^{\mathcal{A}}P_{\mathcal{TA}}]_{ij}
	\\&=&\sum_{m\in\mathcal{T}} F^{\mathcal{A}}_{im}p_{mj} \nonumber
	\\&=&\sum_{m\in\mathcal{T}} (F^{\mathcal{S}}_{im}-\frac{F^{\mathcal{S}}_{ij} F^{\mathcal{S}}_{jm}}{F^{\mathcal{S}}_{jj}})p_{mj} \nonumber
	\\&=&\sum_{m\in\mathcal{T}} F^{\mathcal{S}}_{im}p_{mj}-\frac{F^{\mathcal{S}}_{ij}}{F^{\mathcal{S}}_{jj}}\sum_{m\in\mathcal{T}} F^{\mathcal{S}}_{jm}p_{mj} \nonumber
	\\&=&F^{\mathcal{S}}_{ij}-\frac{F^{\mathcal{S}}_{ij}}{F^{\mathcal{S}}_{jj}}(F^{\mathcal{S}}_{jj}-1) \nonumber
	\\&=&\frac{F^{\mathcal{S}}_{ij}}{F^{\mathcal{S}}_{jj}}, \nonumber
	\end{eqnarray}
	where the third and fifth equalities follow directly from of Theorem~\ref{Nlemma1} and Lemma~\ref{Nlemma2general}, respectively.
\end{proof}

We are now in a position to gather and derive a number of useful relations among the random walk
metrics.

\begin{relation}[Complementary relation of absorption probabilities]
	\begin{equation} \label{Qcomplement}
	Q^{\mathcal{A}}_{ij}=1-\sum_{k\in\mathcal{A}\setminus\{j\}} Q^{\mathcal{A}}_{ik},
	\end{equation}
	where $i \in \mathcal{T}$ and $j \in \mathcal{A}$.
\end{relation}
\begin{proof}
	Based on the definition of $Q$ and the assumption that all the
        nodes in $\mathcal{T}$ are transient, the probability that a random walk eventually ends up in set $\mathcal{A}$ is 1.
\end{proof}

\begin{relation}[Relations between the fundamental matrix and commute time]
	\begin{eqnarray} \label{NpC}
	 &(1)& F_{ii}^{\{j\}}=\pi_i C_{ij} \label{F=piC}
\\&(2)& \frac{F_{im}^{\{j\}}}{\pi_m}+\frac{F_{mi}^{\{j\}}}{\pi_i}=C_{ij}+C_{jm}-C_{im}
\\&(3)& \frac{F_{im}^{\{j\}}}{\pi_m}+\frac{F_{ij}^{\{m\}}}{\pi_j}=C_{jm}
\\&(4)& F_{im}^{\{j\}}+F_{jm}^{\{i\}}=\pi_m C_{ij}
	\end{eqnarray}
\end{relation}
\begin{proof}
	Use (\ref{N+}) and (\ref{C+}).
\end{proof}

\begin{relation}[The hitting time detour overhead in terms of other metrics]
	\begin{eqnarray} 
	 &(1)& H_i^{\{j\}}+H_j^{\{m\}}-H_i^{\{m\}}=\frac{F_{im}^{\{j\}}}{\pi_m} \label{NpH}
\\&(2)& H_i^{\{j\}}+H_j^{\{m\}}-H_i^{\{m\}}= Q_{i}^{\{m,\overline{j}\}}C_{mj}
	\end{eqnarray}
\end{relation}
\begin{proof}
	For the first equation use (\ref{N+}) and (\ref{H+}), and for the second one use the previous equation along with (\ref{Nlemma3}) and (\ref{F=piC}).
\end{proof}

\begin{relation}[The hitting time for two target nodes in terms of
  hitting time for a single target]
	\begin{equation}
		 H_i^{\{j,k\}}=H_i^{\{k\}}-Q_i^{\{j,\overline{k}\}}H_j^{\{k\}}=H_i^{\{j\}}-Q_i^{\{k,\overline{j}\}}H_k^{\{j\}},
	\end{equation}
	which can also be reformulated as: $H_i^{\{j\}}=H_i^{\{j,k\}}+Q_i^{\{k,\overline{j}\}}H_k^{\{j\}}$.
\end{relation}
\begin{proof}
	Aggregate two sides of Eq.(\ref{Nlemma1}) over $m$ and
        substitute Eq.(\ref{Nlemma3}) in it.
\end{proof}

\begin{relation}[Inequalities for hitting time]
	\begin{eqnarray} 
	 &(1)& H_i^{\{m\}}+H_m^{\{j\}} \geq H_i^{\{j\}}  \qquad \textrm{\textbf{(triangular inequality)}} \label{Htrieq}
\\&(2)& H_i^{\{j\}} \geq H_i^{\{j,m\}}
\\&(3)& H_i^{\{m\}}+H_m^{\{j,k\}} \geq H_i^{\{j,k\}}
	\end{eqnarray}
\end{relation} 
\begin{proof}
	For the first inequality, use (\ref{H+}) and (\ref{Lgeq}).
	For the second inequality,  use the aggregated form of
        Eq.(\ref{Nlemma1}) over $m$ and the fact that the entries of
        $F$  are  non-negative. The third inequality  is a
        generalization of the first one.
\end{proof}

\begin{relation}[Inequalities for the fundamental matrix]
	\begin{eqnarray} \label{Ngeq}
	 &(1)& F_{im}^{\{j\}}F_{kk}^{\{j\}}\geq F_{ik}^{\{j\}}F_{km}^{\{j\}}
\\&(2)& F_{kk}^{\{j\}} \geq F_{ik}^{\{j\}}
	\end{eqnarray}
\end{relation}
\begin{proof}
	For the first inequality,  use Eq.(\ref{Nlemma1}) and the
        fact that $F$ is non-negative. The second one can be derived
        from Eqs.(\ref{F=piC}), (\ref{NpH}) and (\ref{Htrieq}). Note
        that these two inequalities hold for any absorbing set
        $\mathcal{A}$, hence we drop the superscripts. 
\end{proof}

\begin{relation}[Inequality for absorption probabilities]
	\begin{equation}
		 Q_{i}^{\{m,\overline{j}\}}\geq Q_{i}^{\{k,\overline{j}\}}Q_{k}^{\{m,\overline{j}\}}
	\end{equation}
\end{relation}
\begin{proof}
	Use (\ref{Nlemma3}) and (\ref{Ngeq}).
\end{proof}

\begin{relation}[Inequality for the digraph Laplacian matrix]
	\begin{equation} \label{Lgeq}
	L_{im}^++L_{kk}^+\geq L_{ik}^++ L_{km}^+
	\end{equation}
\end{relation}
\begin{proof}
	Use (\ref{N+}) and the fact that $F$'s entries are always non-negative.
\end{proof} 

\begin{relation}[Relations for undirected networks (reversible Markov chain)]
	\begin{eqnarray}
	 &(1)&\frac{F_{im}^{\{S\}}}{\pi_m}=\frac{F_{mi}^{\{S\}}}{\pi_i}
\\&(2)&Q_{i}^{\{m,\overline{j}\}}C_m^{\{j\}}= Q_{m}^{\{i,\overline{j}\}}C_i^{\{j\}}
\\&(3)&H_i^{\{m\}}+H_m^{\{j\}}+H_j^{\{i\}}=H_m^{\{i\}}+H_j^{\{m\}}+H_i^{\{j\}}
	\end{eqnarray}
\end{relation}
\begin{proof}
	The first equation follows from Eq.(\ref{N+}) and the fact
        that $L^+$ is symmetric for undirected networks.
	The second equation can be derived by using Eqs.
        (\ref{Nlemma3}), (\ref{N+}), (\ref{C+}) and  the fact that $L^+$ is symmetric.
	The third equation follows from Eq.(\ref{H+}) and  $L^+$ being symmetric.
\end{proof}

\begin{figure}
	\centering
	\includegraphics[width=0.8\textwidth]{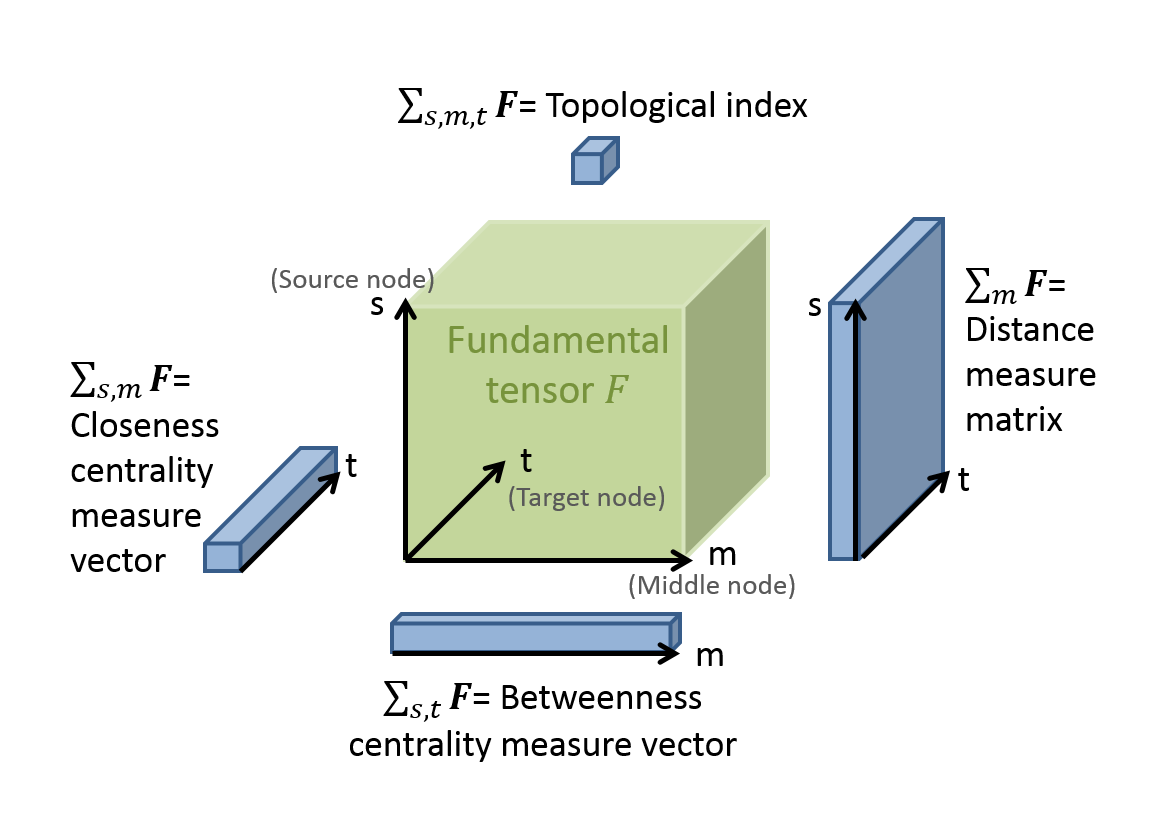}
	\caption{Markov fundamental tensor and unifying framework for computing the random walk distance, betweenness measure, closeness measure, and topological index}\label{fig:tensor}
\end{figure}



\section{Unifying Random-Walk Distance, Centrality, and Topological Measures} \label{sec:measureUnification}

Many network measures have been proposed in the literature for network
analysis purposes \cite{borgatti2006graph}, such as distance metrics for measuring the similarity (or diversity) between nodes or
entities of a network, centrality measures to assess a node's involvement or importance
in the connectivity or communication between network entities, and
topological indices to measure the structural stability of
networks. In this section, we review some of these network measures
proposed in the literature, and show that these measures can be
unified in terms of the fundamental tensor, which  provides a
coherent framework for computing them and understanding the relations
among them. 
\begin{statement}
	Fundamental tensor $\boldsymbol{F}$ unifies various network
        random-walk measures via summation along one or more dimensions shown in Figure~\ref{fig:tensor}.
\end{statement}

\subsection{Random-walk distance measure}
The hitting time metric has been used extensively in different
application domains, such as a distance (or dissimilarity) measure for clustering and
classification purposes~\cite{chen2008clustering}. Note that this
distance metric satisfies two out of three conditions for the general
distance metric: It is positive when two ends are different and zero
when two ends are identical. As noted earlier, the hitting time metric
is in general not symmetric, but it satisfies the triangle inequality.
In Section~\ref{sec:definition}, we have shown that hitting time can
be computed from the fundamental tensor by summing over $m$'s (the
middle node dimension, see Figure~\ref{fig:tensor}).
\begin{equation}
Distance_{\text{rw}}(s,t) = H_s^{\{t\}}=\sum_m \boldsymbol{F}_{smt}.
\end{equation}
With a cost matrix $W$, the hitting cost distance (\ref{eq:B=Fr}) is
obtained by the weighted sum over the medial node dimension of the fundamental tensor: $\U_s^{\{t\}}= \sum_m \boldsymbol{F}_{smt} b_m,$
where $b_m = \sum_i w_{mi}p_{mi}$ is the expected outgoing cost of node $m$.  

\subsection{Random-walk centrality measures}
Network centrality measures can be  broadly categorized into two main types~\cite{borgatti2006graph}: i)
distance-based and ii) volume-based. The {\em closeness} centrality is an
example of the  distance-based measures, whereas the  {\em
  betweenness} centrality is an example of volume-based measures. 
\begin{itemize}
	\item Random-walk closeness measure:  Closeness centrality
          ranks nodes in a network  based on their total distance from
          other nodes of the network. This measure reflects how
          easily/closely the node is accessible/reachable from the
          other parts of the network,  and in a nutshell how
          ``central'' the node  is located within a network. The classical
          closeness centrality metric is defined using the shortest
          path distances.  Noh and
          Rieger \cite{noh2004random} introduces the {\em random walk
            closeness centrality}, which is  defined using the hitting
          time distance: A node is considered to have a high
          centrality value if and only if  its total hitting time distances
          from other nodes in the network is small. This closeness
          centrality measure  can be easily expressed in terms of the random walk fundamental tensor:
	\begin{equation}
	Closeness_{\text{rw}}(t)=\sum_s H_s^{\{t\}}=\sum_{s,m} \boldsymbol{F}_{smt},
	\end{equation}
	or in the reciprocal form to imply lower importance with small closeness value: $Closeness_{\text{rw}}(t)=\frac{|\mathcal{V}|}{\sum_{s,m} \boldsymbol{F}_{smt}}$.
	
	\item Random-walk betweenness measure: Betweenness measure
          quantifies the number of times a node acts as a ``bridge''
          along the paths between different parts of the network. 
          The larger the number of paths crossing that node, the more
          central the node is. As a special case, the node degree, $deg(m)$, can
          be viewed as a betweenness centrality measure in an
          undirected network. Clearly, it captures how many paths of
          length 1 going through node $m$ (or many 1-hop neighbors it
          has)~\cite{borgatti2006graph}. It is also proportional to
          the total number of (random) walks passing through node $m$
          from any source to any target in the network. This follows
          from the following more general statement. For a general
          (strongly connected)  network, we define the
          {\em random walk betweenness} of node $m$ as follows and
          show that it is proportional to $\pi_m$, the stationary
          probability of node $m$:
	\begin{eqnarray} \label{eq:betweenness}
	Betweenness_{\text{rw}}(m) &=& \sum_{s,t} \boldsymbol{F}_{smt}
	\\&=& \sum_{s,t} (L_{sm}^+ - L_{tm}^+ - L_{st}^+ + L_{tt}^+)\pi_m  
	\\&=& |\mathcal{V}| \sum_t L_{tt}^+\pi_m
	\\&=& |\mathcal{E}| K\pi_m,
	\end{eqnarray} 
	where $K$ is the Kirchhoff index (see
        Section~\ref{sec:Kirchhoff}). The third equality follows by
        using the fact that  the column sum of the digraph Laplacian
        matrix $L^+=(\Pi(I-P))^+$ is zero
        \cite{Li:Digraph,boley2011commute}. For a (connected)
        undirected network, $\pi_m=\frac{d_m}{2|\mathcal{E}|}$, where $d_m$ is the
        degree of node $m$.

	For undirected networks, Newman \cite{newman2005measure} proposes  a variation of the random
        walk betweenness measure defined above, which we denote by
        $Betweenness_{\text{Newman,bidirect}}(m)$ (the use of
        subscript {\em bidirect} will be clear below): it is defined
        as the (net) electrical  current flow $I$ through a medial
        node in an undirected network (which can be viewed as an
        electrical resistive network with bi-directional links with
        resistance),  when a unit current flow is injected at a source
        and removes at a target (ground), aggregated over all such
        source-target pairs. Formally, we have 
	\begin{eqnarray} 
	Betweenness_{\text{Newman,bidirect}}(m) &=& \sum_{s,t} I(s\rightarrow m\rightarrow t) \nonumber
	\\&=& \sum_{s,t} \sum_k \frac{1}{2}|\boldsymbol{F}_{smt} p_{mk} - \boldsymbol{F}_{skt} p_{km}|. \nonumber
	\end{eqnarray}
We remark that the original definition given by Newman is based on
current flows in electrical networks, and is only valid for {\em
  undirected} networks. Define $f(\boldsymbol{F}_{smt}) = \sum_k
\frac{1}{2}|\boldsymbol{F}_{smt} p_{mk} - \boldsymbol{F}_{skt}
p_{km}|$, then 
$Betweenness_{\text{Newman,directed}}(m)=\sum_{s,t}f(\boldsymbol{F}_{smt})$
yields a general definition of Newman's random walk betweenness measure that also
holds for directed networks. In particular, we show that if a network
is strictly {\em unidirectional}, namely, if $e_{ij}\in E$ then
$e_{ji}\notin E$, Newman's random walk betweenness centrality reduces
to $	Betweenness_{\text{rw}}(m)=|\mathcal{E}| K\pi_m$: 
	\begin{eqnarray}
	Betweenness_{\text{Newman,unidirect}}(m) &=& \sum_{s,t} \sum_{k} |\boldsymbol{F}_{smt} p_{mk}| \nonumber
	\\&=& \sum_{s,t} \boldsymbol{F}_{smt} \sum_{k} p_{mk} \nonumber
	\\&=& \sum_{s,t} \boldsymbol{F}_{smt} = |\mathcal{E}| K\pi_m, \label{eq:sumTensor}
	\end{eqnarray}
	where $K$ is the Kirchhoff index (see
        Section~\ref{sec:Kirchhoff}) and the last equality follows
        from Corollary~\ref{cor:F}. 
\end{itemize}

\subsection{Kirchhoff Index}\label{sec:Kirchhoff}
The term {\em topological index} is a single metric that characterizes
the topology (``connectivity structure'') of a network; it has been
widely used in mathematical chemistry to reflect certain structural
properties  of the underlying molecular graph
\cite{rouvray1986predicting}\cite{rouvray1985role}.
Perhaps the most known topology index is the Kirchhoff
index~\cite{klein1993resistance} which has found a variety of
applications~\cite{zhou2008note,palacios2010bounds,bendito2010kirchhoff,ranjan2013geometry,palacios2001resistance}. Kirchhoff index is also closely connected to Kemeney's constant \cite{wang2017kemeny,kirkland2016random}. 
The Kirchhoff index is often defined in terms of {\em effective
  resistances}~\cite{klein1993resistance}, $K(G)=\frac{1}{2}\sum_{s,t}
\Omega_{st}$, which is closely related to commute times, as
$\Omega_{st}=\frac{1}{|\mathcal{E}|}C_{st}$~\cite{tetali1991random}.
Hence we have 
\begin{equation}
K(G)=\frac{1}{2|\mathcal{E}|}\sum_{s,t} C_{st}= \frac{|\mathcal{V}|}{|\mathcal{E}|}\sum_t L_{tt}^+=\frac{1}{|\mathcal{E}|} \sum_{s,m,t} \boldsymbol{F}_{smt}, 
\end{equation}
where the second equality comes  from Eq.(\ref{C+}). 
In other words, the Kirchhoff index can be computed by summing over
all three dimensions in Figure~\ref{fig:tensor}, normalized by the
total number of edges. 

The relations between Kirchhoff index, effective resistance, and Laplacian matrix have been well studied in the literature.
 The authors
in~\cite{ranjan2013geometry} provided three interpretations of
$L^+_{ii}$ as a topological centrality measure, from effective resistances in an electric network, random
walk detour costs, and graph-theoretical topological connectivity
via connected bi-partitions, and demonstrate that the Kirchhoff index, as a
topological index, captures the overall robustness of a network.
 The relation between the effective resistance and the Moore-Penrose inverse of the Laplacian matrix is more elaborated in \cite{gutman2004generalized}, and insightful relations between Kirchhoff index and inverses of the Laplacian eigenvectors can be found in \cite{zhu1996extensions,gutman1996quasi}.

\section{Characterization of Network Articulation Points and Network Load Distribution}\label{sec:articulpoint}
We extend the definition of {\em articulation point} to a general (undirected and directed) network as a node whose
removal reduces the amount of reachability in the network. For
instance, in a network $G$, if $t$ is previously reachable
from $s$, i.e. there was at least one path from $s$ to $t$, but $t$ is no longer reachable
from $s$ after removing $m$, node $m$ is an
articulation point for network $G$. Note that $s$ may still be reachable from $t$ after removing $m$ in a directed network, which is not the case for an undirected network. Hence, in an undirected network, the reduction in the number of reachabilities results to the increase in the number of connected components in the network, which is the reason to call articulation point as {\em cut vertex} in the undirected networks. Removal of an articulation point in a {\em directed} network,
however, does not necessarily increase the number of connected components in the
network. 

As an application of the fundamental tensor, we introduce  the normalized fundamental tensor
$\hat{\boldsymbol{F}}$ and show that its entries contain information
regarding articulation points in a general (directed or undirected)
network.  If $\boldsymbol{F}_{smt}$ exists, its {\em normalized}
version is defined as follow,
\begin{equation}\label{eq:normalized-fundamental-tensor}
\hat{\boldsymbol{F}}_{smt}=\begin{cases}
\frac{\boldsymbol{F}_{smt}}{\boldsymbol{F}_{mmt}}  & \text{if } s,m\neq t \\
0  & \text{if } s=t \text{ or } m=t
\end{cases} 
\end{equation}
The normalized fundamental tensor satisfies the following properties: a)
$0\leq \hat{\boldsymbol{F}}(s,m,t) \leq 1$, and b)
$\hat{\boldsymbol{F}}_{smt} = Q_s^{\{m,\overline{t}\}}$.
Recall that $Q_s^{\{m,\overline{t}\}}$ is the absorption probability
that a random walk starting from node $s$ hits (is absorbed by) node
$m$ sooner than node $t$.  The second property (b) is a result of  Theorem~\ref{Nlemma3}  and the
first property (a) follows from (b). Clearly,
$\hat{\boldsymbol{F}}_{smt} =Q_s^{\{m,\overline{t}\}}=1$ means that
with probability $1$, any random walk starts from node $s$ always hit
node $m$ before node $t$. Hence node $m$ is on any path (thus walk)
from $s$ to $t$. Hence it is an articulation point. We therefore have
the following statement:
\begin{statement}
The normalized fundamental tensor captures the articulation points of
a network: if $\hat{\boldsymbol{F}}_{smt}=1$, then node $m$ is an
articulation point; namely, node $m$ is located on all paths from $s$ to $t$. On the other extreme, $\hat{\boldsymbol{F}}_{smt}=0$ indicates
that $m$ is not located on any path from $s$ to $t$ and thus it plays no role for this reachability.  
\end{statement}

Figure~\ref{fig:ArticulationPoint} depicts  two simple networks, one
undirected and one directed, and displays the corresponding normalized
fundamental tensors we have computed for these two networks (the
tensors are ``unfolded'' as a series of matrices, each with fixed
$t$). Any column that contains an entry with value $1$ indicates the
corresponding node $m$, $1 \leq m \leq 5$, is located on all paths between a pair of source and target, and so is an articulation point
for the network \footnote{As a convention, the source node is considered as the articulation point of the reachability, but not the target.}. Counting the number of 1's in each column $m$ over the
entire tensor yields the number of source-target pairs for which  node $m$ is
an articulation point. The larger this count is, the more critical
node $m$ is for the overall network reachability. For instance, for
both networks,  node 3 is the most critical node for network reachabilities. 

\begin{figure}
	\centering
	\includegraphics[width=0.6\textwidth, height=0.9\textwidth]{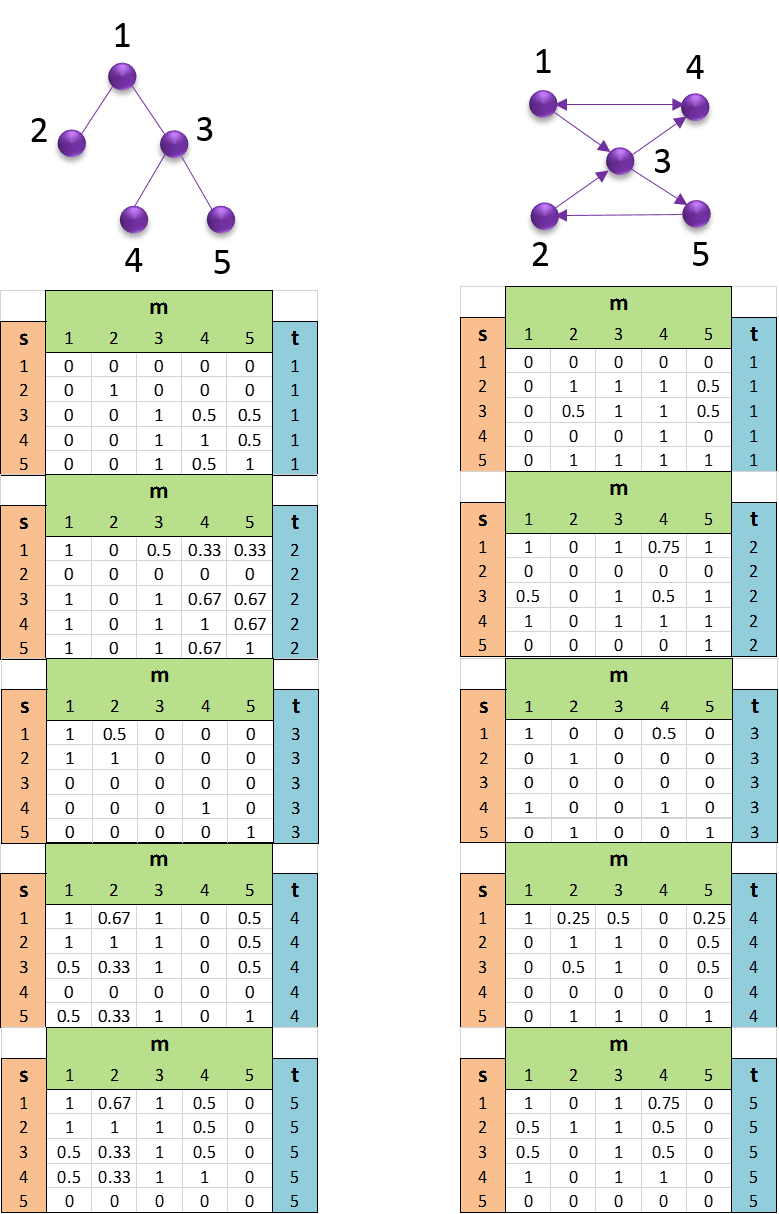}
	\caption{Two networks, one undirected and one directed, and the corresponding normalized fundamental tensor}\label{fig:ArticulationPoint}
\end{figure}

More generally, we can view $\hat{\boldsymbol{F}}_{smt}$ as a measure
of how critical a role node $m$ plays in the reachability from node
$s$ to node $t$. As a generalization of articulation points, we define
the overall {\em load} that node $m$ carries for all source-target
pairs in a network as follows:
\begin{equation} \label{eq:load}
Load(m) = \frac{1}{(n-1)^2}\sum_{s,t} \hat{\boldsymbol{F}}_{smt}, 
\end{equation}

It is interesting to compare Eq.(\ref{eq:load}) with
Eq.(\ref{eq:betweenness}), where the latter (the unnormalized 
summation  $\sum_{s,t} \boldsymbol{F}_{smt}$) is proportional to the
stationary probability of node $m$ (and degree of $m$ if the network
is undirected). The distribution of $Load(m)$'s provides a
characterization of how balanced a network in terms of distributing its load
(reachability between pairs of nodes), or how robust it is against
targeted attacks. A network with a few high-valued articulation points
({\em e.g.}, a star network) is more vulnerable to the failure of a few
nodes. Using a few synthetic networks with specific topologies
as well as real-world networks as examples,  Figure~\ref{fig:load}
plots the distribution of $Load(m)$ for these networks (sorted based
on increasing values of $Load(m)$'s).  Among the specific-shaped networks, it is interesting to note that
comparing to a chain network, the loads on a cycle network are evenly
distributed -- this shows the huge difference that adding a single
edge can make in the structural property of a network. It is not
surprising that the complete graph has evenly distributed loads. 
In constrast, a star graph has the most skewed load distribution, with
the center node as an articulation point of the network. Comparing the
binary tree, the grid network has a less skewed load distribution. It
is also interesting to compare the load distribution of the binary
with that of  a  3-ary ``fat tree'' network -- such a topology is used
widely in data center networks~\cite{fat_tree}. The real networks used
in Figure~\ref{fig:load}(b) include the Arxiv High Energy Physics
 - Phenomenology collaboration network (CAHepPh)~\cite{ca-HepPh},
 a sampled network of Facebook~\cite{facebook}, the coauthorship
 network of network scientists (netSci)~\cite{netSci}, the Italian
 power grid~\cite{powerItaly}, and a protein-protein  interaction network~\cite{protein}.
 For comparison, we also include three networks generated via two
 well-known random network models,  the Preferential Attachment generative model (PA)
 \cite{PA} and  Erdos Renyi (ER) random graph model \cite{random} with two different initial links of 8
 (random) and 40 (random2). We see that the two ER random networks
 yield most balanced load distributions, whereas the PA network
 exhibits behavior similar to a tree network, with a few nodes bearing
 much higher loads than others. The real networks exhibit load
 distributions varying between these types of random networks (with
 the  Italian power grid closer to an ER random network, whereas
 netSci closer to a PA random network).

\begin{figure}%
	\centering
	\subfloat[]{{\includegraphics[width=0.45\textwidth]{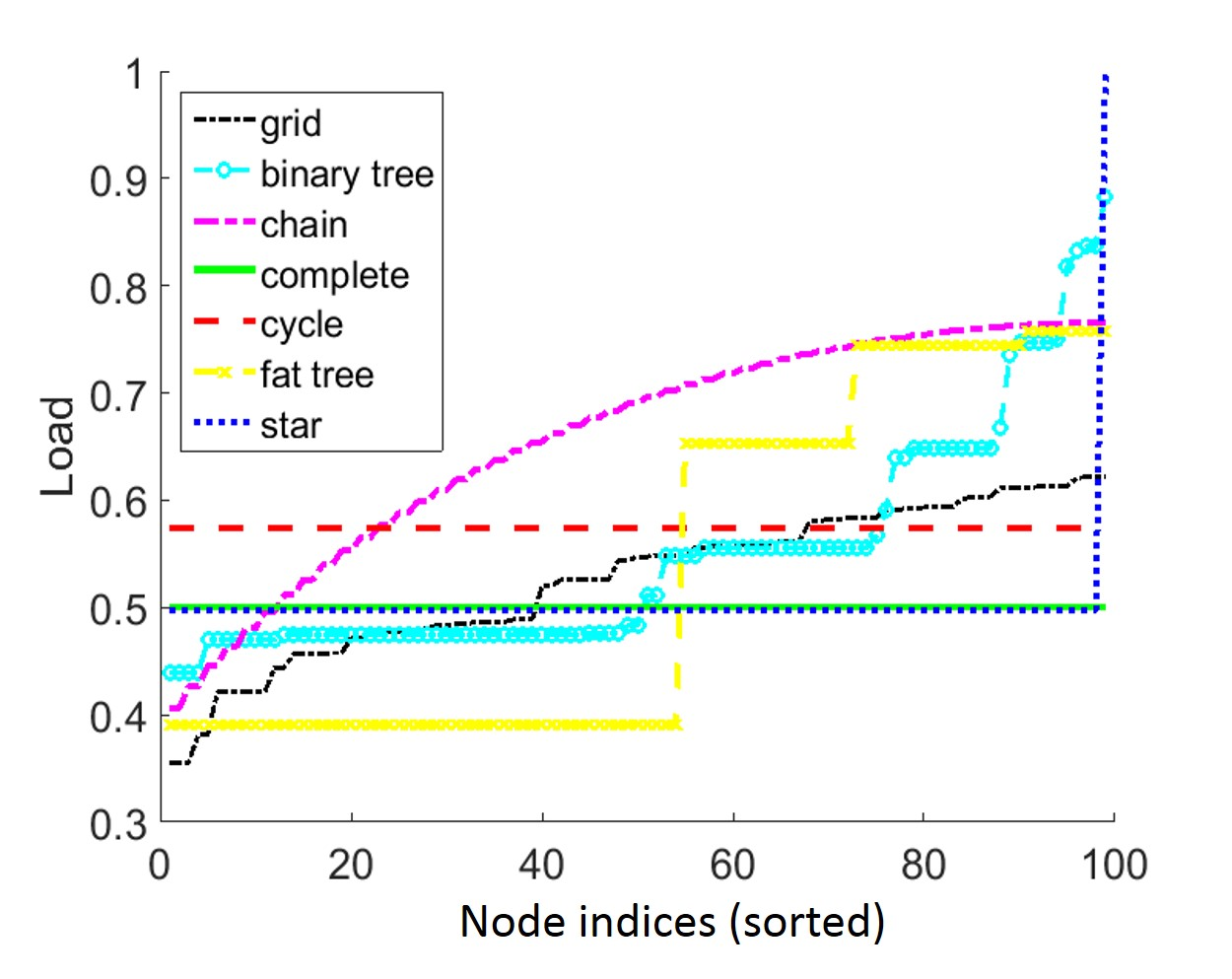} }}%
	~
	\subfloat[]{{\includegraphics[width=0.45\textwidth]{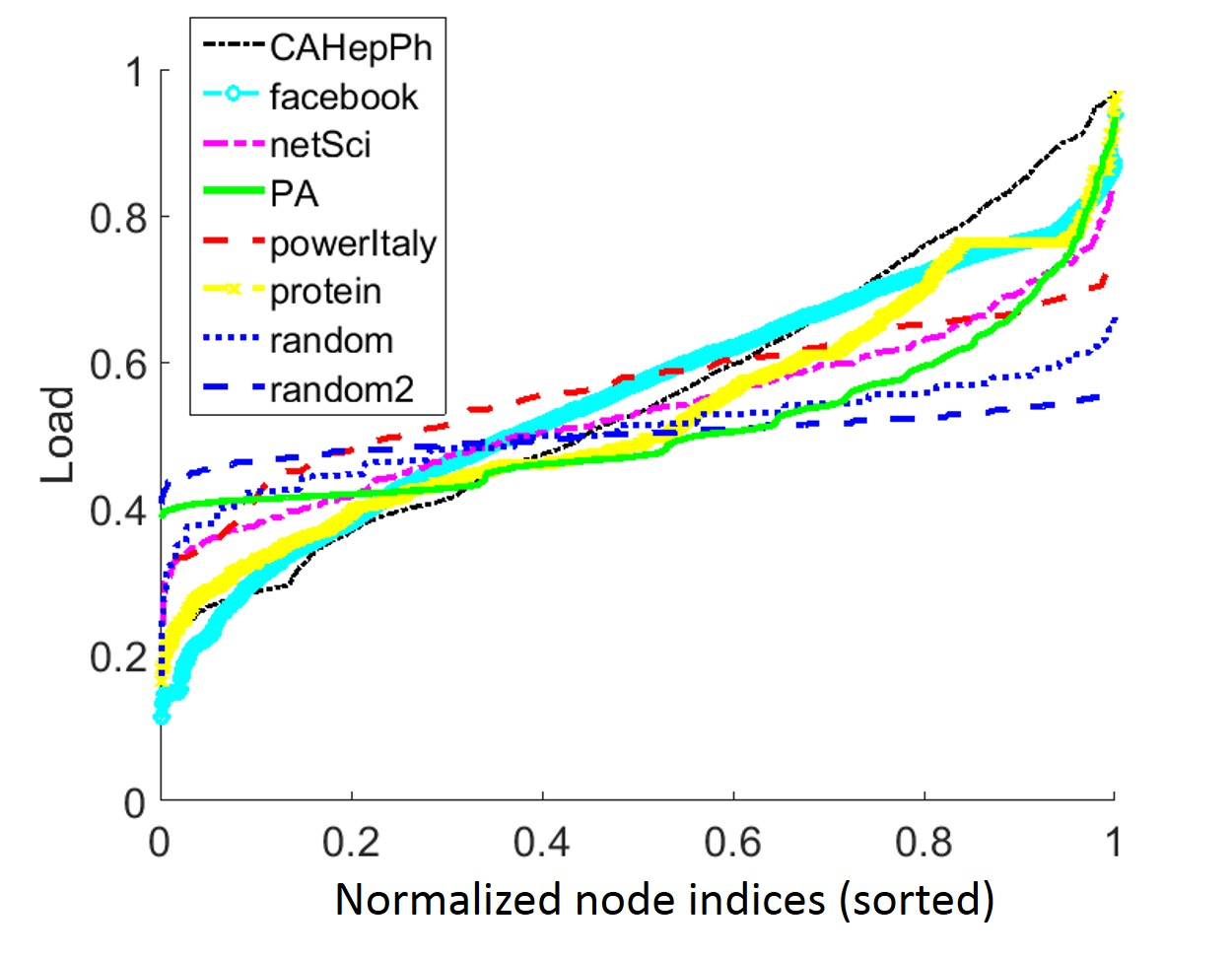} }}%
	\caption{Load balancing in a) specific-shaped networks and b) real-world networks}\label{fig:load}
\end{figure} 

\section{Most Influential Nodes in Social Networks} \label{sec:most-influential}
Online social networks have played a key role as a medium for the
spread of information, ideas, or ``influence'' among their
members. The Influence maximization problem in social networks is about finding the most
influential persons who can maximize the spread of influence in the
network. This problem has applications in viral marketing, where a
company may wish to spread the publicity, and eventually the adoption, of a new
product via the most influential persons in popular social networks. A
social network is modeled as a (directed or undirected) graph where nodes represent the
users, and edges represent relationships and interactions between the
users. An {\em influence cascade} over a network can be modeled by a
diffusion process, and the objective of the influence maximization
problem is to find the $k$ most influential persons as the initial
adopters who will lead to most number of adoptions.  

The heat conduction model \cite{golnari2015revisiting} is a diffusion
process which is inspired by how heat transfers through a medium from the part with higher temperature to the part with lower
temperature. In this diffusion process, the probability that a user 
adopts the new product is a linear function of adoption probabilities 
of her friends who have influence on her as well as her own
independent tendency. We modeled the independent tendency of users
for the product adoption by adding an {\em exogenous} node, indexed as $o$, and linked to all of the
nodes in the network. Network $G$ with added node $o$ is called
extended $G$, denoted by $G^o$. We showed that the influence
maximization problem for $k=1$, where $k=\# initial\text{ }adopters$,
under the heat conduction diffusion process has the following solution in
terms of the normalized fundamental tensor over $G^o$
\cite{golnari2015revisiting}:    
\begin{equation} \label{eq:MIP}
t^* = \argmax_{t} \sum_{s \in \mathcal{V}} \hat{\boldsymbol{F}}_{sto}.
\end{equation}

We also proved that the general influence maximization problem for $k>1$ is NP-hard \cite{golnari2015revisiting}. However, we proposed an
efficient greedy algorithm, called {\em C2Greedy}
\cite{golnari2015revisiting}, which finds a set of
initial adopters who produce a provably near-optimal 
influence spread in the network. The algorithm iteratively finds the
most influential node using Eq.(\ref{eq:MIP}), then removes it from the
network and solves the equation to find the next best initiator.  

\begin{statement} For $k=1$,  $\argmax_{t} \sum_{s \in \mathcal{V}} \hat{\boldsymbol{F}}_{sto}$
finds the most influential node of network $G^o$ as the initial
adopter for maximizing the influence spread over the network with heat
conduction \cite{golnari2015revisiting} as the diffusion process.
For $k>1$, the greedy algorithm,  {\em C2Greedy}~\cite{golnari2015revisiting}, employs this relation to
iteratively find the $k$ most influential nodes, which yields a
provably near-optimal solution. 
\end{statement}

\begin{figure}%
	\centering
	\subfloat[]{{\includegraphics[width=0.45\textwidth]{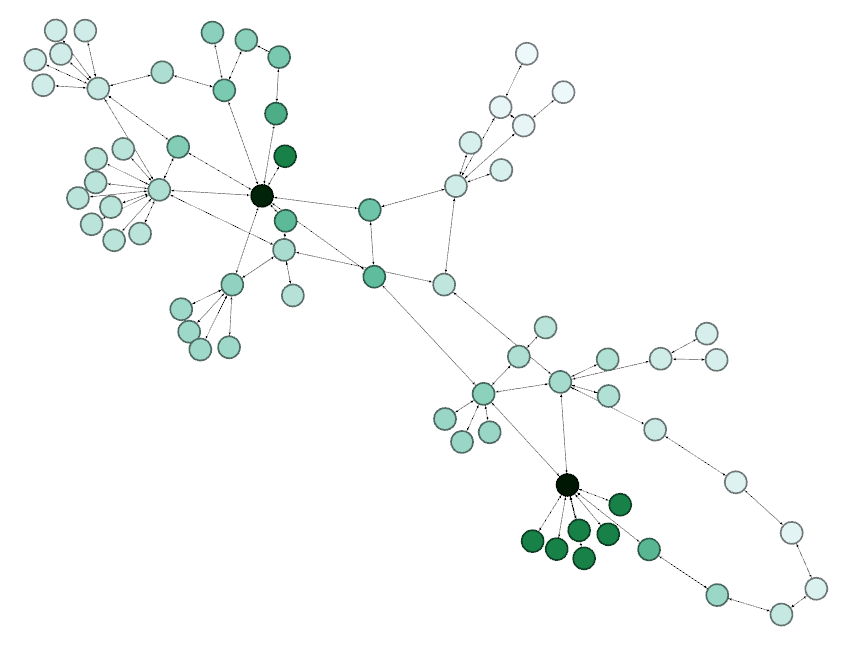} }}%
	~
	\subfloat[]{{\includegraphics[width=0.45\textwidth]{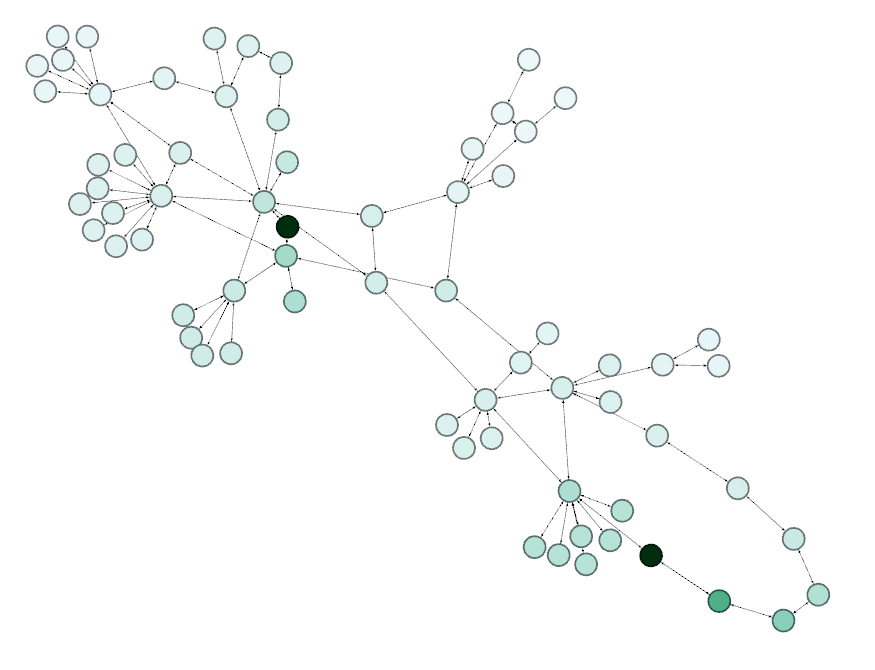} }}%
	\caption{Influence spread by a) two most influential nodes found from C2Greedy~\cite{golnari2015revisiting}, b) two neighbors of nodes in part a}%
	\label{fig:influence_spread}%
\end{figure}

In \cite{golnari2015revisiting}, we showed that C2Greedy outperforms the state-of-the-art influence maximization algorithms in both performance and speed, which we do not repeat here. Instead, we present two new sets of experiments in the rest.
 
We remark that the metric in Eq.(\ref{eq:MIP}) addresses both the
{\em global} characteristics of the network by placing the most
influential node in the critical and strategical ``center'' of the network,
and the {\em local} characteristics by specifying the highly populated
and ``neighbor-rich'' nodes. In Figure~\ref{fig:influence_spread}(a), we 
visualize the influence spread of
the two most influential nodes found from C2Greedy using the ESNet \cite{esnet} network. The initiators are colored
in black, and the green shades indicate the
influence spread over the nodes in the network; the darker the green, the higher probability of production adoption for the node. In
Figure~\ref{fig:influence_spread}(b), we pick two nodes, which are a neighbor of the two
most influential nodes identified by C2Greedy, as the initiators, and
visualize the probability of influence spread caused by these two
nodes over other nodes in the network. The lower green intensity of
Figure~\ref{fig:influence_spread}(b), compared to that of Figure~\ref{fig:influence_spread}(a)
shows that not any two initiators -- even if they are their immediate neighbors and {\em
  globally} located very closely -- can cause the same influence spread
as the  two most influential nodes identified by C2Greedy.

Moreover, we show that the $k$ most influential initiators found by C2Greedy have higher influence spread in the network compared to that of well-known centrality/importance measure algorithms: 1- top $k$ nodes with highest
(in-)degree, 2- top $k$ nodes with highest closeness centrality scores, 3- top $k$ nodes with highest
Pagerank score~\cite{pagerank}, and 4- a benchmark which consists of $k$ nodes picked randomly. For this purpose, we use real-world network data from three social
networks, wiki vote \cite{wikiVote}, hepPh citation \cite{cit-HepPh},
and Facebook \cite{facebook}. Figure~\ref{fig:k_init} illustrates how the C2Greedy outperforms the other algorithms for a wide range of $k$.

\begin{figure}%
	\centering
	\subfloat[wiki vote]{{\includegraphics[width=0.3\textwidth]{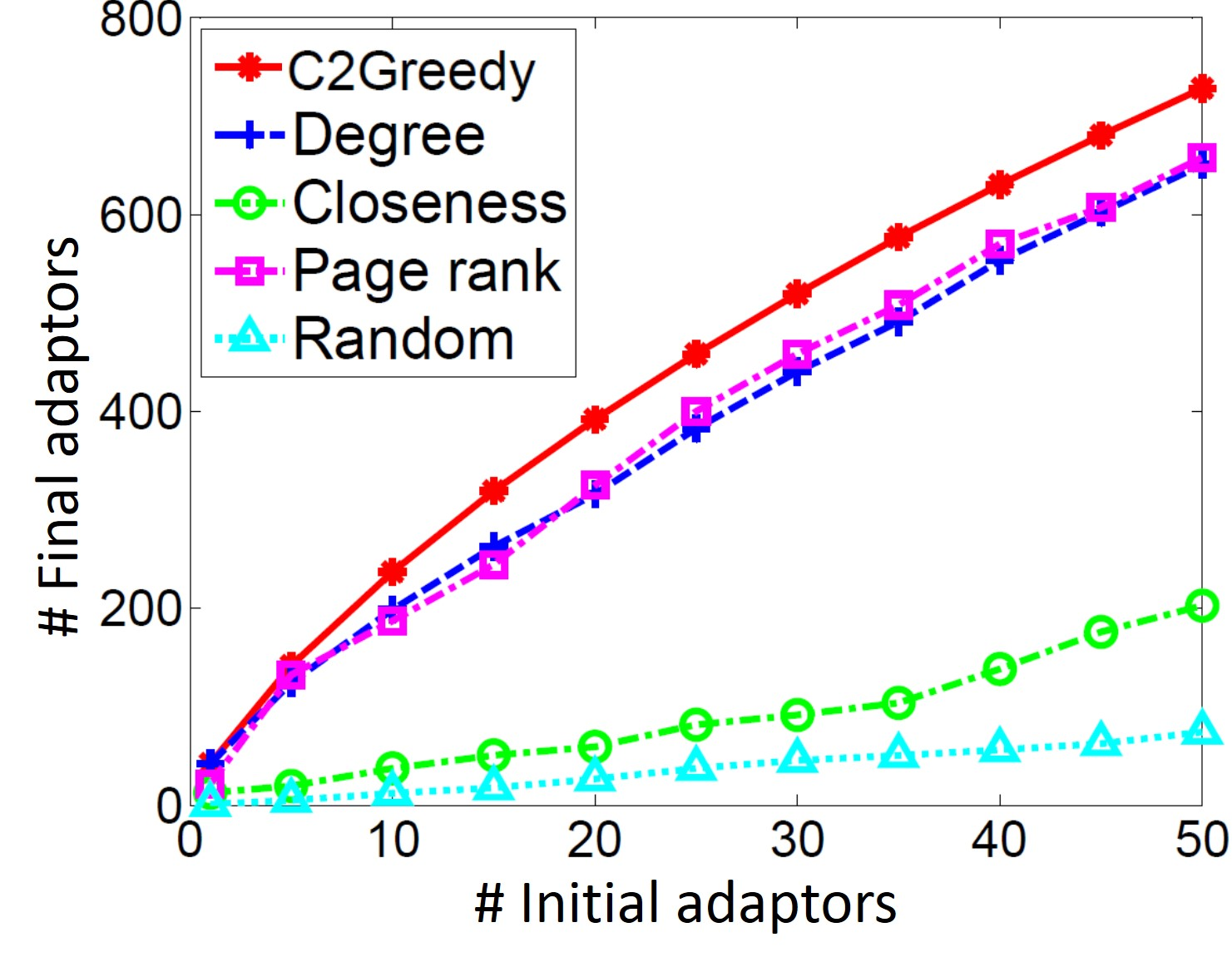} }}%
	~
	\subfloat[hepPh citation]{{\includegraphics[width=0.3\textwidth]{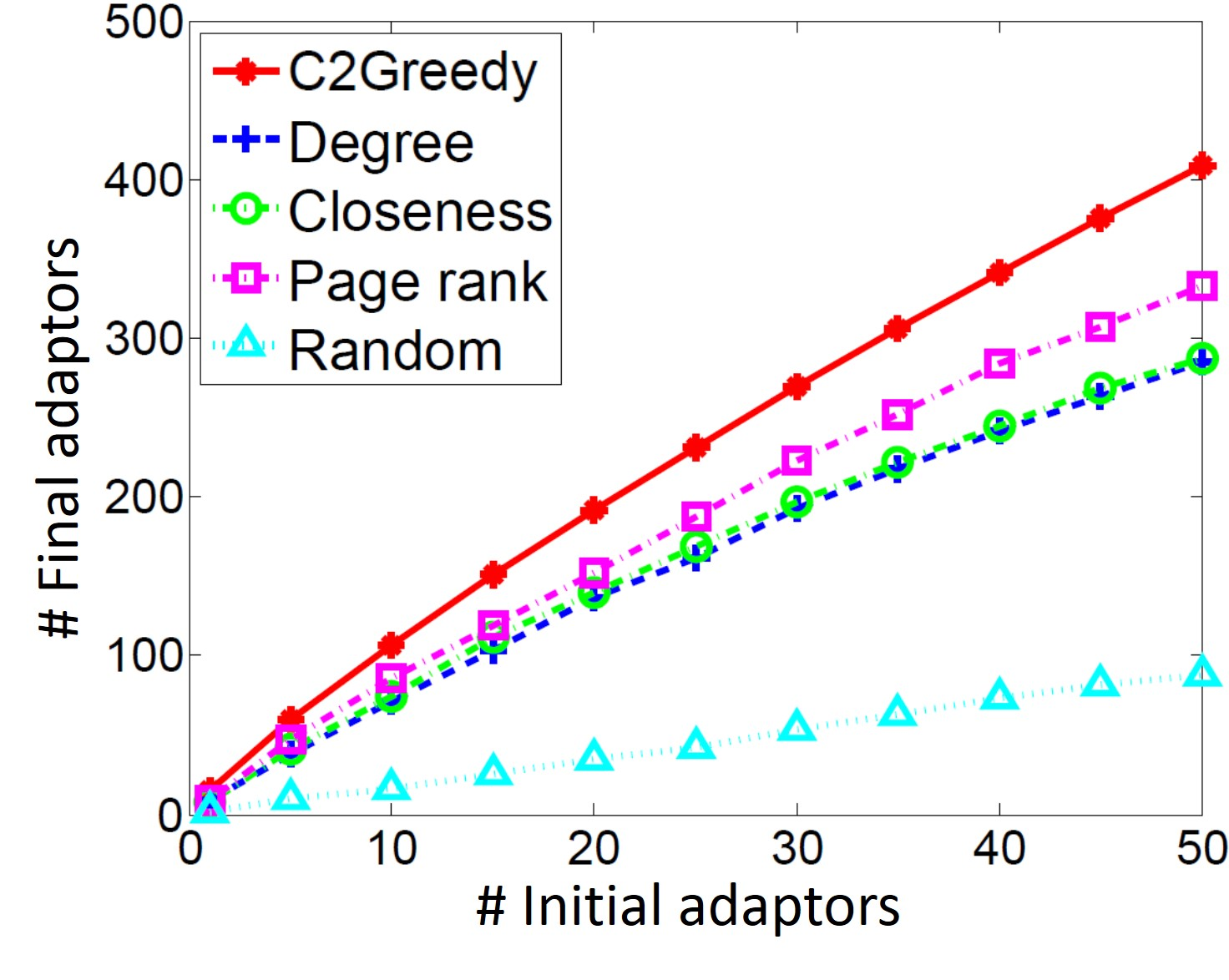} }}%
	~
	\subfloat[facebook]{{\includegraphics[width=0.3\textwidth]{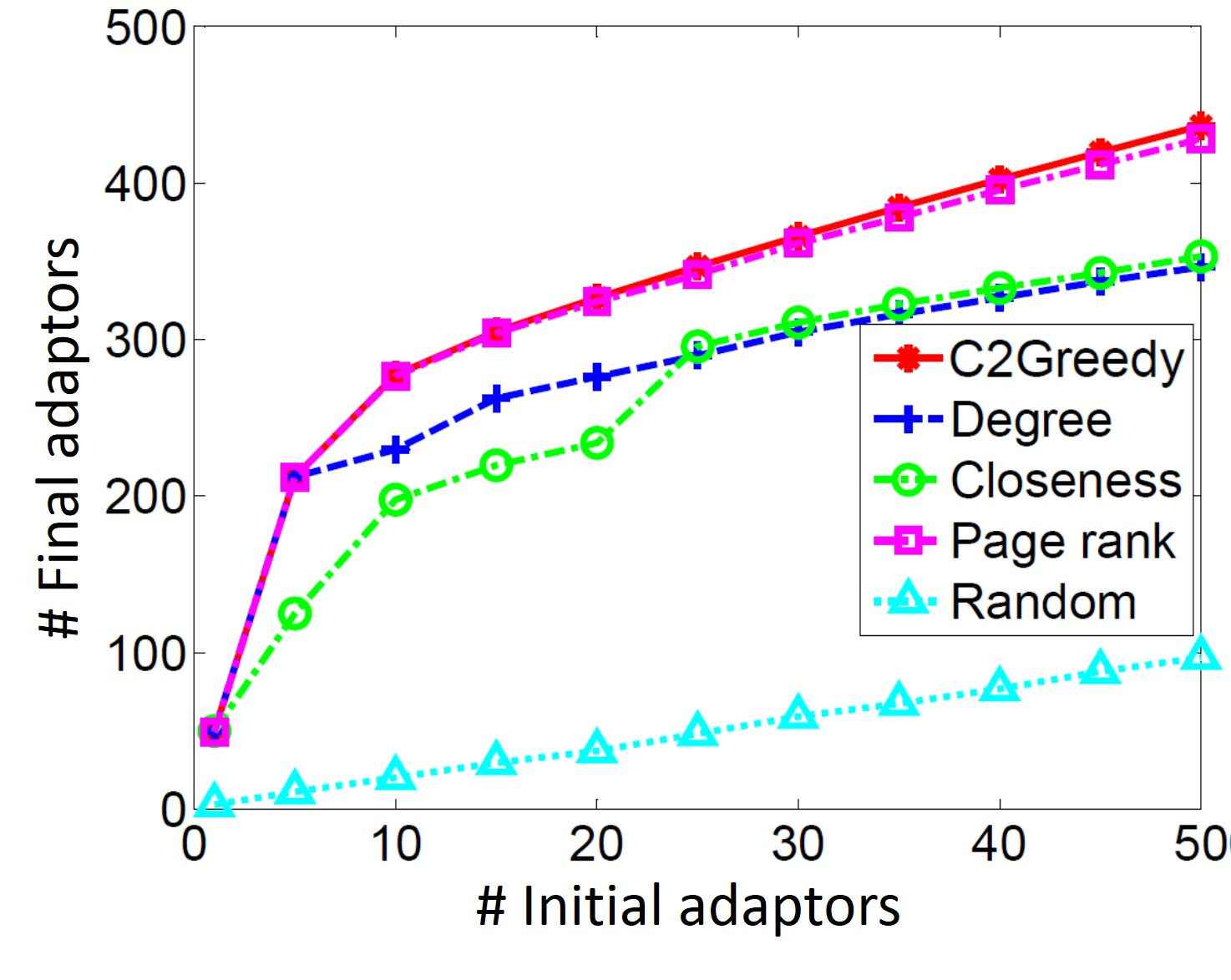} }}%
	\caption{Influence spread comparison between 5 different ``most influential nodes" methods for a wide choice of initiator sizes.}%
	\label{fig:k_init}%
\end{figure}

\section{Fast Computation of Dynamic Network Reachability Queries}\label{sec:reachability}
Reachability information between nodes in a network is crucial for a
wide range of applications  from gene interactions in bioinformatics
\cite{van2000representing} to XML query processing
\cite{chamberlin2002xquery}. Given a network, a reachability query $R(s,
t)$ has a binary answer with 1 indicating that target node $t$ is
reachable from source node $s$, and 0 representing that it is
not.  Several efficient algorithms have been devised to answer reachability queries when the
network is {\em static}~\cite{chen2008efficient,merz2014preach,jin2013simple,wang2006dual}.
However, few efficient solutions have been developed to answer
reachability query for {\em dynamic} networks, {\em e.g.}, after node or
link failures. For example,  garbage collection in programming
languages requires dynamic (re-)computation of reachability  to balance the reclamation of memory, which might be
reallocated. The speed of answering reachability queries affect the performance of  applications~\cite{michael2002safe}. 

As a final application of the fundamental tensor, we illustrate how it can
be employed to develop an efficient algorithm to answer reachability
queries for dynamic networks. Here we do {\em not} require the network $G$
under consideration (before or after failures) be strongly connected,
otherwise the reachability query problem is trivial. For simplicity of
exposition, in the following we only consider node failures. Similar in
Section~\ref{sec:most-influential}, we add an exogenous node
$o$ to network $G$ and connecting all the nodes to it. We note that
this  extended network $G^o$ has only one recurrent equivalence class,
and $\boldsymbol{F}_{sto}$ exists for any pairs of $s$ and $t$.  Moreover,
 $\boldsymbol{F}_{sto}$ is non-zero if and only $t$ is reachable from
 $s$ in $G$. This is because with non-zero probability a random walk will  visit every node that is reachable
 from $s$  before hitting $o$. By pre-computing the fundamental matrix
 $F^{\{o\}}$ once, we can answer any reachability query
 $R(s,t)$ in constant time using $F^{\{o\}}$ by  performing a table look-up.

Now suppose a set of nodes, $\mathcal{F}$, fail. We claim that we can
answer the {\em dynamic} reachability query $R(s,t,\mathcal{F})$ (after
the nodes in $\mathcal{F}$ fail, but {\em without} prior knowledge of the
node failure set $\mathcal{F}$)  in $O(|\mathcal{F}|)$. In particular,
if $|\mathcal{F}|$ is of a constant order $O(1)$ compared to the size
of network $|\mathcal{V}|$, then the queries are answered in constant $O(1)$
time.  This is achieved by leveraging Theorem~\ref{Nlemma1} for incremental computation of the
fundamental matrix. Let 
 $\mathcal{S}=\mathcal{F}\cup\{o\}$ and define $\boldsymbol{F}_{st\mathcal{S}}$
\begin{equation}
\boldsymbol{F}_{st\mathcal{S}}=\boldsymbol{F}_{sto}-\boldsymbol{F}_{s\mathcal{F}o}\boldsymbol{F}^{-1}_{\mathcal{F}\mathcal{F}o}\boldsymbol{F}_{\mathcal{F}to},
\end{equation}
which is the tensor form of $F^{\mathcal{S}}_{st} =
F^{\{o\}}_{st}-F^{\{o\}}_{s\mathcal{F}}(F^{\{o\}}_{\mathcal{F}\mathcal{F}})^{-1}F^{\{o\}}_{\mathcal{F}t}$. Note
that the sub-matrix $(F^{\{o\}}_{\mathcal{F}\mathcal{F}})^{-1}$  is
non-singular. This comes from the fact that $F^{\{o\}}$ is an
inverse M-matrix (an inverse M-matrix is a matrix whose
inverse is an M-matrix), hence each of its principal sub-matrix is also an
inverse M-matrix. We have the following statement:

\begin{statement} \label{stm:reachability}
	In the extended network $G^o$, $\boldsymbol{F}_{sto}$ is
        non-zero if and only if $t$ is reachable from $s$ in the
        original network $G$.  Furthermore,
 if the nodes in the set $\mathcal{F}$ fail,
 $\boldsymbol{F}_{st\mathcal{S}}$ is non-zero (where $\mathcal{S}=\mathcal{F}\cup\{o\}$) if and only if $t$ is
 still reachable from $s$ in network $G$ after the failures.
\end{statement}
Using the above statement and Theorem~\ref{Nlemma1}, we can answer
({\em static} and {\em dynamic}) reachability queries both before and after
failures in constant times (for a constant size node failure set
$\mathcal{F}$) by pre-computing $\boldsymbol{F}_{::o}$ (=$F^{\{o\}})$
and storing the results in a table.  The method for answering
reachability queries is summarized in Algorithm (\ref{alg1}). The function
${1}_{\{b\}}$ is an indicator function which is equal to 1 if
$b=True$ and 0 if $b=False$.

\begin{algorithm}                      
	\caption{ \sc Answering a reachability query }
	\label{alg1}                           
	\begin{algorithmic}[1]
		\small                    
		\STATE {\bfseries query:} $R(s,t,\sim\mathcal{F})$
		\STATE {\bfseries input:} transition matrix $P$ of the extended network $G^o$
		\STATE {\bfseries precomputed oracle:} $\boldsymbol{F}_{::o}=(I-P_{\setminus o})^{-1}$
		\STATE {\bfseries output:} answer to reachability queries. 
		\IF{$\mathcal{F}= \emptyset$} \STATE{$R(s,t)={1}_{\{\boldsymbol{F}_{sto}>0\}}$} \ELSE \STATE{$R(s,t,\sim\mathcal{F})={1}_{\{\boldsymbol{F}_{sto}-\boldsymbol{F}_{s\mathcal{F}o}\boldsymbol{F}^{-1}_{\mathcal{F}\mathcal{F}o}\boldsymbol{F}_{\mathcal{F}to}>0\}}$} \ENDIF
	\end{algorithmic}
\end{algorithm}

\section{Conclusion}
We revisited the fundamental matrix in Markov chain theory and extended it to the fundamental tensor which we showed that can be built much more efficiently than computing the fundamental matrices separately ($O(n^3)$ vs. $O(n^4)$) for the applications that the whole tensor is required. We also showed that fundamental matrix/tensor provides a unifying framework to derive other Markov metrics and find useful relations in a coherent way. We then tackled four interesting network analysis applications in which fundamental tensor is exploited to provide effective and efficient solutions: 
1) we showed that fundamental tensor unifies various network random-walk measures, such as distance, centrality measure, and topological index, via summation along one or more dimensions of the tensor; 2) we extended the definition of articulation points to the directed networks and used the (normalized) fundamental tensor to compute all the articulation points of a network at once. We also devised a metric to measure the load balancing over nodes of a network. Through extensive experiments, we
evaluated the load balancing in several specifically-shaped networks and real-world networks; 3) we showed that (normalized) fundamental tensor can be exploited to infer the cascade and spread of a phenomena or an influence in social networks. We also derived a formulation to find the most influential nodes for maximizing
the influence spread over the network using the (normalized) fundamental tensor, and demonstrated the efficacy of our method compared to other well-known ranking methods through multiple real-world network experiments; and 4) we presented a
dynamic reachability method in the form of a pre-computed oracle which is cable of answering
to reachability queries efficiently both in the case of having failures or no failure
in a general directed network.

\section*{Acknowledgement}
The research was supported in part by  US DoD DTRA grants HDTRA1-09-1-0050
and HDTRA1-14-1-0040, and ARO MURI Award W911NF-12-1-0385. We would also like to thank our colleagues, Amir Asiaei and Arindam Banerjee, who helped us with the results reported in Section \ref{sec:most-influential}.


\bibliographystyle{elsarticle-num} 
\bibliography{main-zhili}


%
%
%

\newpage
\section{Appendix} \label{sec:appendix}
In this appendix, we provide the detailed derivations regarding the
relations between the stochastic form and matrix form of hitting time
and hitting costs, respectively.

$\bullet$ \textbf{ Relation between the stochastic form and matrix form of hitting time}

Let $t$ be the only absorbing node and rest of nodes belong to $\mathcal{T}$, then:
\begin{eqnarray}
H_s^{\{t\}}&=&\sum_{k=1} k\sum_{m\in\mathcal{T}} {[P^{k-1}_{\mathcal{TT}}]}_{sm}{[P_{\mathcal{TA}}]}_{mt}=\sum_{k=1} k\sum_{m\in\mathcal{T}} {[P^{k-1}_{\mathcal{TT}}]}_{sm}(1-\sum_{j\in\mathcal{T}}{[P_{\mathcal{TT}}]}_{mj})  \nonumber
\\&=& \sum_{k=1} k(\sum_{m\in\mathcal{T}} {[P^{k-1}_{\mathcal{TT}}]}_{sm} - \sum_{j\in\mathcal{T}}{[P^{k}_{\mathcal{TT}}]}_{sj}) = \sum_{m\in\mathcal{T}}\sum_{k=1} k({[P^{k-1}_{\mathcal{TT}}]}_{sm} - {[P^k_{\mathcal{TT}}]}_{sm}) \nonumber
\\&=& \sum_{m\in\mathcal{T}} {[P^{k-1}_{\mathcal{TT}}]}_{sm} = \sum_{m} F^{\{t\}}_{sm},
\end{eqnarray}
which is the matrix form of hitting time Eq.(\ref{eq:H}).

$\bullet$ \textbf{ Relation between the stochastic form and matrix form of hitting cost}

Let $\mathcal{Z}_{sm}$ be the set of all possible walks from $s$ to $m$
and $\zeta_j$ be the $j$-th walk from this set. We use
$\mathcal{Z}_{sm}(l)$ to denote the subset of walks whose total length
is $l$, and $\mathcal{Z}_{sm}(k,l)$ to specify the walks which have
total length of $l$ and total step size of $k$. Recall that a walk (in
contrast to a path) can have repetitive nodes, and the length of a
walk is the sum of the edge weights  in the walk and its step size is
the number of edges. Recall that
$\mathbb{P}(\eta_t=l|X_0=s)$ denotes the probability of hitting $t$ in
total length of $l$ when starting from $s$, which can be obtained from
the probability of walks: $\mathbb{P}(\eta_t=l|X_0=s)=\sum_{\zeta_j\in
  Z_{st}(l)} \textrm{Pr}_{\zeta_j}$. Probability of walk $\zeta_j$
denoted by $\textrm{Pr}_{\zeta_j}$ is computed by the production over
the  probabilities of passing edges:
$\textrm{Pr}_{\zeta_j}=p_{sv_1}p_{v_1v_2}...p_{v_{k-1}m}$, where
$p_{vu}$ is the edge probability from $v$ to $u$. The summation over
the walk probabilities is computed using the following relation: 
\begin{equation} \label{eq:Pr_walk}
\sum_{\zeta_j\in Z_{sm}(k)} \textrm{Pr}_{\zeta_j}=
\begin{cases} 
[P_{\mathcal{TT}}^k]_{sm} & \text{if } m\in\mathcal{T} \\
[P_{\mathcal{TT}}^{k-1}P_{\mathcal{TA}}]_{sm} & \text{if } m\in\mathcal{A}
\end{cases}
\end{equation}
With this introduction, the derivation of the stochastic form of
hitting cost Eq.(\ref{eq:stoch_B}) can proceed as follows:

\begin{eqnarray}
\U_s^{\{t\}}&=&\sum_{l\in\mathcal{C}}l\sum_{k=1}^{<\infty}\sum_{\zeta_j\in \mathcal{Z}_{st}(k,l)} \textrm{Pr}_{\zeta_j} 
\\&=& \sum_{l\in\mathcal{C}}\sum_{k=1}^{<\infty}\sum_{\zeta_j\in \mathcal{Z}_{st}(k,l)}l_{\zeta_j} \textrm{Pr}_{\zeta_j} \nonumber
\\&=& \sum_{\zeta_j\in \mathcal{Z}_{st}}l_{\zeta_j} \textrm{Pr}_{\zeta_j}
\\&=&\sum_{\zeta_j\in \mathcal{Z}_{st}} \textrm{Pr}_{\zeta_j} \sum_{k=1}^{k_{\zeta_j}} w_{v_{k-1}v_k} 
\\&=& \sum_{\zeta_j\in \mathcal{Z}_{st}}\sum_{k=1}^{k_{\zeta_j}}[\prod_{i=1}^{k}p_{v_{i-1}v_i}.(p_{v_kv_{k+1}}w_{v_kv_{k+1}}).\prod_{i=k+2}^{k_{\zeta_j}}p_{v_{i-1}v_i}]
\\&=& \sum_{e_{xy}\in E}p_{xy}w_{xy}(\sum_{\zeta_j\in \mathcal{Z}_{sx}}\textrm{Pr}_{\zeta_j}).(\sum_{\zeta_i\in \mathcal{Z}_{yt}}\textrm{Pr}_{\zeta_i}) \label{eq:multiplicationPrinciple}
\\&=& \sum_{e_{xy}\in E}p_{xy}w_{xy}(\sum_k\sum_{\zeta_j\in \mathcal{Z}_{sx}(k)}\textrm{Pr}_{\zeta_j}).(\sum_k\sum_{\zeta_i\in \mathcal{Z}_{yt}(k)}\textrm{Pr}_{\zeta_i})
\\&=& \sum_{e_{xy}\in E}p_{xy}w_{xy}(\sum_k [P_{\mathcal{TT}}^k]_{sx}).(\sum_k [P_{\mathcal{TT}}^{k-1}P_{\mathcal{TA}}]_{yt}) \label{eq:table}
\\&=& \sum_{e_{xy}\in E}p_{xy}w_{xy}F^{\{t\}}_{sx} Q^{\{t\}}_y
\\&=& \sum_{e_{xy}\in E}p_{xy}w_{xy}F^{\{t\}}_{sx} \label{eq:Q=1}
\\&=& \sum_x F^{\{t\}}_{sx} \sum_{y\in \mathcal{N}_{out}(x)} p_{xy}w_{xy}
\\&=& \sum_x F^{\{t\}}_{sx} r_x,
\end{eqnarray}
where $l_{\zeta_j}$ and $k_{\zeta_j}$ denote the length and step size
of a  walk $\zeta_j$, respectively, and $r_x=\sum_{y\in
  \mathcal{N}_{out}(x)} p_{xy}w_{xy}$ is the average outgoing cost of
node $x$. In the above derivation,  Eq.(\ref{eq:table}) comes from
Eq.(\ref{eq:Pr_walk}), and Eq.(\ref{eq:Q=1}) follows from the fact that $Q^{\{t\}}_y=1$ when having
$t$ as the only absorbing node in the network and reachable from all
the other nodes.

\end{document}